\documentclass[a4paper]{article}
\pdfoutput=1
\usepackage[pagebackref,colorlinks=true,linkcolor=blue,urlcolor=blue,citecolor=blue,pdfstartview=FitH]{hyperref}
\usepackage{amsmath,amsfonts,amssymb,amsthm}
\usepackage{fullpage,color}
\usepackage{enumitem,ifthen}

\usepackage{setspace}
\usepackage{mathpazo}
\usepackage{nicefrac}

\usepackage[nameinlink,capitalize]{cleveref}
\renewcommand{\eqref}[1]{\hyperref[#1]{(\ref*{#1})}}

\usepackage[colorinlistoftodos]{todonotes}
\usepackage{bbm}

\newtheorem{theorem}{Theorem}[section]
\newtheorem{lemma}[theorem]{Lemma}
\newtheorem{corollary}[theorem]{Corollary}
\newtheorem{claim}[theorem]{Claim}
\newtheorem{proposition}[theorem]{Proposition}
\newtheorem{definition}[theorem]{Definition}
\theoremstyle{definition}
\newtheorem{remark}[theorem]{Remark}
\newtheorem{conjecture}[theorem]{Conjecture}

\renewcommand{\epsilon}{\varepsilon}
\renewcommand{\phi}{\varphi}
\newcommand{\eps}{\epsilon}
\DeclareMathOperator*{\E}{\mathbb{E}}

\providecommand{\RR}{\mathbb{R}}
\providecommand{\integers}{\mathbb{Z}}
\providecommand{\Dist}[1]{\Pi_{#1}}
\providecommand{\Distv}{\vec{\Pi}}
\newcommand{\calD}{{\mathcal{D}}}
\DeclareMathOperator{\im}{im}

\DeclareMathOperator*{\EE}{\mathbb{E}}
\DeclareMathOperator{\Gr}{Gr}
\newcommand{\R}{\mathbb{R}}
\newcommand{\iprod}[1]{ \langle{#1}\rangle}
\newcommand{\norm}[1]{\left \| {#1} \right \|}

\newcommand{\one}{\mathbbm{1}}
\newcommand{\abs}[1]{\left | {#1} \right |}
\newcommand{\sqbinom}[2]{\begin{bmatrix}#1\\#2\end{bmatrix}}

\newcommand{\set}[1]{\left \{ {#1} \right \} }
\newcommand{\sett}[2]{\left\{{#1} \;\middle|\; {#2}\right\}}

\newcommand\Tstrut{\rule{0pt}{2.6ex}}         
\newcommand\Bstrut{\rule[-0.9ex]{0pt}{0pt}}   

\usepackage[colorinlistoftodos]{todonotes}

\newif\ifcomments
\commentsfalse

\ifcomments
\newcommand{\prahladhuvacha}[1]{\todo[color=red!100!green!33,inline,size=\small]{ph:
    #1}}

\newcommand{\inote}[1]{\todo[color=blue!20,size=\footnotesize]{Irit: #1}}
\newcommand{\ynote}[1]{\todo[color=red!40!blue!15,size=\footnotesize]{Yuval: #1}}
\newcommand{\phnote}[1]{\todo[color=red!100!green!33, size=\footnotesize]{ph: #1}}
\newcommand{\dnote}[1]{\todo[color=green, size=\footnotesize]{yotam: #1}}
\else
\newcommand{\prahladhuvacha}[1]{}
\newcommand{\inote}[1]{}
\newcommand{\ynote}[1]{}
\newcommand{\phnote}[1]{}
\newcommand{\dnote}[1]{}
\fi

\title{Boolean function analysis on high-dimensional expanders\thanks{A preliminary version of this paper appeared in {\em Proc. 20th RANDOM}, 2018~\cite{DiksteinDFH2018}.}}
\author{Yotam Dikstein\thanks{Weizmann Institute of Science, ISRAEL. email: {\tt \{yotam.dikstein,irit.dinur\}@weizmann.ac.il}. The research of the first and second authors was supported in part by Irit Dinur's ERC-CoG grant 772839.}
\and
Irit Dinur\footnotemark[2]
\and Yuval Filmus\thanks{Technion --- Israel Institute of Technology, ISRAEL. email: {\tt yuvalfi@cs.technion.ac.il}. Taub Fellow --- supported by the Taub Foundations. The research was funded by ISF grant 1337/16.}
 \and Prahladh Harsha\thanks{Tata Institute of Fundamental Research, INDIA. email: {\tt prahladh@tifr.res.in}. Research supported in part by UGC--ISF grant. Part of the work was done when the author was visiting the Weizmann Institute of Science.}}

\begin{document}

\maketitle

\begin{abstract}

We initiate the study of Boolean function analysis on high-dimensional expanders. We give a random-walk based definition of high-dimensional expansion, which coincides with the earlier definition in terms of two-sided link expanders. Using this definition, we describe an analog of the Fourier expansion and the Fourier levels of the Boolean hypercube for simplicial complexes. Our analog is a decomposition into approximate eigenspaces of random walks associated with the simplicial complexes. Our random-walk definition and the decomposition have the additional advantage that they extend to the more general setting of posets, encompassing both high-dimensional expanders and the Grassmann poset, which appears in recent work on the unique games conjecture.

We then use this decomposition to extend the Friedgut-Kalai-Naor theorem to high-dimensional expanders. Our results demonstrate that a constant-degree high-dimensional expander can sometimes serve as a sparse model for the Boolean slice or hypercube, and quite possibly additional results from Boolean function analysis can be carried over to this sparse model. Therefore, this model can be viewed as a derandomization of the Boolean slice, containing only $|X(k-1)|=O(n)$ points in contrast to $\binom{n}{k}$ points in the $(k)$-slice (which consists of all $n$-bit strings with exactly $k$ ones).

\end{abstract}
\section{Introduction}
Boolean function analysis is an essential tool in theory of computation. Traditionally, it studies functions on the Boolean cube $\{-1,1\}^n$. Recently, the scope of Boolean function analysis has been extended further, encompassing groups~\cite{EllisFF2015a,EllisFF2015b,Plaza2015,EllisFF2017}, association schemes~\cite{ODonnellW2013,Filmus2016-fkn,Filmus2016-orthogonal,FilmusM2019,FilmusKMW2018,DinurKKMS2018-grassman,KhotMS2018}, error-correcting codes~\cite{BarakGHMRS2015}, and quantum Boolean functions~\cite{MontanaroO2010}. Boolean function analysis on extended domains has led to progress in learning theory~\cite{ODonnellW2013} and on the unique games conjecture~\cite{KhotMS2017,DinurKKMS2018-grassman,DinurKKMS2018-2to1,BarakKS2019,KhotMS2018}.

High-dimensional expanders emerged in recent years as a new area of study, of interest to several different communities. Just as expander graphs are sparse models of the complete graph, so are high-dimensional expanders sparse models of a higher-dimensional object, namely the \emph{complete hypergraph}. Expander graphs are  central objects, appearing in a diverse list of areas. High-dimensional expanders are much newer objects which have already found connections to topological overlapping theory~\cite{FoxGLNP2012,Evra2017}, to analysis of Markov chains~\cite{AnariLOV2019}, and to coding theory~\cite{DinurHKNT2021} and property testing~\cite{KaufmanL2014}. (Note that while expander graphs are explicit derandomizations of random graphs, the mere existence of high-dimensional expanders is surprising since there is no sparse random model for generating these objects.).

\smallskip
\emph{The goal of this work is to connect these two threads of research, by introducing Boolean function analysis on high-dimensional expanders.}
\smallskip

We study Boolean functions on simplicial complexes. A \emph{pure $d$-dimensional simplicial complex $X$} is a set system consisting of an arbitrary collection of sets of size $d+1$ together with all their subsets. The sets in a simplicial complex are called \emph{faces}, and it is standard to denote by $X(i)$ the faces of $X$ whose cardinality is $i+1$. Our simplicial complexes are \emph{weighted} by a probability distribution $\Dist{d}$ on the top-level faces, which induces probability distributions $\Dist{i}$ on $X(i)$ in a natural way for all $i$: we choose $s \sim \Dist{d}$, and then choose an $i$-face $t \subset s$ uniformly at random. Our main object of study is the space of functions $f \colon X(d)\to\R $, and in particular, Boolean functions $f\colon X(d)\to\{0,1\}$.

\subsection{Many different definitions of high-dimensional expansion}
Although graph expansion has many definitions, all of which are equivalent, they each generalize to higher dimensions differently, leading to a diverse landscape of definitions.

\smallskip

The first definition studied by Linial and Meshulam \cite{LinialM2006} and by
Gromov \cite{Gromov2010} was topological, focusing on generalizations of edge expansion through coboundary maps in higher dimensions.
It was later discovered that certain bounded-degree simplicial complexes constructed by Lubotzky, Samuels and Vishne \cite{LubotzkySV2005-exphdx,LubotzkySV2005-hdx} satisfy this definition (more accurately, a variant of it called cosystolic expansion), leading to the first known family of bounded-degree complexes that satisfy Gromov's so-called ``topological overlap property'' \cite{KaufmanKL2016,DotterrerKW2018,EvraK2016}. The LSV construction comes from arithmetic quotients of Bruhat--Tits buildings, and are high-dimensional generalizations of the celebrated Lubotzky--Philips--Sarnak construction of Ramanujan expander graphs \cite{LubotzkyPS1988}.

\smallskip

Subsequent works were interested in additional properties exhibited by the LSV complexes (and others, see \cite{Li2004}) that aren't necessarily captured by the topological definitions mentioned above. For example, Dinur and Kaufman~\cite{DinurK2017} proved that the LSV complexes support so-called agreement tests that are studied in the context of probabilistically checkable proofs, and were previously known for only dense families of subsets such as the complete complex. The relevant definition for that work is spectral link-expansion, which we now describe.

Let $X$ be a $d$-dimensional simplicial complex, and let $s \in X$ be any face of dimension $\leq d-1$. The graph of the link $s$ is the graph whose vertex set consists of all elements $v \notin s$ such that $\set{v} \cup s \in X$. The edges are all pairs $\set{v,u}$ such that $\set{v,u} \cup s \in X$. A simplicial complex $X$ is a two-sided (or one-sided) \emph{link-expander} with spectral radius $\gamma$ if for every link, the non-trivial normalized eigenvalues are upper-bounded by $\gamma$ in the one-sided case, or sandwiched between $-\gamma$ and $\gamma$ in the two-sided case.

Garland \cite{Garland1973} had studied this type of spectral expansion in links, and used it to show the vanishing of the real cohomology of Bruhat--Tits buildings. Similar techniques were further explored in the work of Oppenheim~\cite{Oppenheim2018}. The notion of one-sided spectral expansion also appeared in earlier works of Kaufman, Kazhdan and Lubotzky~\cite{KaufmanKL2016}, where it was applied towards proving topological expansion.

\smallskip

A third definition, through random walks on the $i$-faces, was studied initially by Kaufman and Mass~\cite{KaufmanM2018}, where the authors defined a
combinatorial ``random-walk'' type of expansion, and showed that this type of expansion was implied by expansion of the links. This notion is concerned with
the convergence speed of high-dimensional random walks to the stationary distribution. This was refined by Dinur and Kaufman~\cite{DinurK2017}, who showed that two-sided link-expansion implies that all random walks on a high-dimensional expander converge at approximately the same speed as on the complete complex, with an error term dominated by $\gamma$.

In this paper we continue to study this two-sided definition, and show that it is in fact \emph{equivalent} to a new random-walk definition which we suggest. We find the new random-walk definition appealing because it is very natural to state, and at the same time equivalent to the powerful two-sided link-expansion definition. Moreover, the random walk definition generalizes naturally beyond simplicial complexes also to ranked posets (partially ordered sets). Finally, the random walk point of view allows for doing an analog of Fourier analysis, as we discuss below.

\smallskip

Concurrently and independently of this work, Kaufman and Oppenheim~\cite{KaufmanO2020} studied the connection between one-sided link-expansion and convergence of the relevant random walks. They showed that one-sided link-expansion (which is weaker than its two-sided variant) is already enough for getting the same conclusions about the speed of convergence of random walks as was shown for the two-sided case. This was picked up in an exciting work by Anari, Liu, Oveis Gharan and Vinzant~\cite{AnariLOV2019}, who relied on this connection to solve a longstanding open question in the area of Markov chain sampling. They showed that a certain well-studied Markov chain on bases of matroids can be viewed as a random walk on the faces of a certain one-sided link expander, thereby using the work of Kaufman and Oppenheim~\cite{KaufmanO2020} to prove convergence of this random walk.
These techniques were further developed in subsequent works that analyzed other Markov chains using their underlying high dimensional expanding structure \cite{AnariLO2024, ChenGSV2021, FengGYZ2022, ChenLV2023, ChenLV2021, AbdolazimiLO2022}.

Two independent works by \cite{BafnaHKL2022} and \cite{GurLL2022} used alternative decompositions of functions on high dimensional expanders to show hypercontractive properties of random walks in high dimensional expanders.

\subsection{Random-walk based definition of high-dimensional expanders}

\medskip

Denote the real-valued function space on $X(i)$ by $C^i := \{f\colon X(i) \to \RR \} $. There are two natural operators $U_i\colon C^i \to C^{i+1}$ and $D_{i+1}\colon C^{i+1} \to C^i$, which are defined by averaging:
\begin{align*}
U_if(s) &:= \EE_{t \sim \Dist{i}}[f(t) | t \subset s] \; \; \left (= \frac{1}{i+2} \sum_{t \subset s} f(t) \right ), \\
D_{i+1}f(t) &:= \EE_{s \sim \Dist{i+1}}[f(s) | s \supset t].
\end{align*}

The compositions $D_{i+1}U_i$ and $U_{i-1}D_i$ are Markov operators of two natural random walks on $X(i)$, the {\em upper random walk} and the {\em lower random walk}.

The first walk we consider is the \emph{upper random walk} $D_{i+1}U_i$. Given a face $t_1 \in X(i)$, we choose its neighbor $t_2$ as follows: we pick a random $s \sim \Dist{i+1}$ conditioned on $s \supset t_1$ and then choose uniformly at random $t_2 \subset s$. Note that there is a probability of $\frac{1}{i+2}$ that $t_1 = t_2$. We define the \emph{non-lazy upper random walk} by choosing $t_2 \subset s$ conditioned on $t_1 \ne t_2$. We denote the Markov operator of the non-lazy upper walk by $M^+_i$.
This operator satisfies the following equality (which can be used to define it)
\[ D_{i+1}U_i = \frac{i+1}{i+2} M_i^+ + \frac 1{i+2}I\,.
\]

A similar (but not identical) random walk on $X(i)$ is the \emph{lower random walk} $U_{i-1}D_i$. Here, given a face $t_1 \in X(i)$, we choose a neighbor $t_2$ as follows: we first choose a $r \in X(i-1), r \subseteq t_1$ uniformly at random and then choose a $t_2 \sim \Dist{i}$ conditioned on $t_2 \supset r$.

For instance, if $X$ is a graph (a $1$-dimensional simplicial complex), then the non-lazy upper random walk is the usual adjacency walk we define on a weighted graph (i.e.\ traversing from vertex to vertex by an edge). The (lazy) upper random walk has probability $\frac{1}{2}$ of staying in place, and probability $\frac{1}{2}$ of going to different adjacent vertex. The lower random walk on $V = X(0)$ doesn't depend on the current vertex: it simply chooses a vertex at random according to the distribution $\Dist{0}$ on $X(0)$.

The Up and Down operators resemble operators in several similar situations. One immediate example is the boundary and coboundary operators with real coefficients. These differ from the operators described above as they include signs according to orientation of the faces, whereas our operators ignore signs.  Stanley studied Up and Down operators in numerous combinatorial situations and the most relevant to this work is his definition of sequentially differential posets, which we discuss in \Cref{sec:intro-posets}.

The Up and Down operators also make an appearance in the Kruskal--Katona theorem. O'Donnell and Wimmer~\cite{ODonnellW2013} related the non-lazy upper and lower random walks (the two random walks are identical in their setting) to the Kruskal--Katona theorem, and used this connection to construct an optimal net for monotone Boolean functions.
\medskip

We are now ready to give our definition of a high-dimensional expander in terms of these walks.

\begin{definition}[High-Dimensional Expander] \label{def:intro-HDX}
Let $\gamma < 1$, and let $X$ be a $d$-dimensional simplicial complex. We say $X$ is a \emph{$\gamma$-high-dimensional expander} (or \emph{$\gamma$-HDX}) if for all $0 \leq i \leq d-1$, the non-lazy upper random walk is $\gamma$-similar to the lower random walk in operator norm in the following sense:
\[ \norm{M^+_i - U_{i-1}D_i} \leq \gamma. \]
\end{definition}

In the graph case, this coincides with the definition of a $\gamma$-two-sided spectral expander: recall that the lower walk on $X(0)$ is by choosing two vertices $v_1,v_2 \in X(0)$ independently. Thus $\norm{M^+_i - U_{i-1}D_i}$ is the second eigenvalue of the adjacency random walk in absolute value. For $i \geq 1$, we cannot expect the upper random walk to be similar to choosing two independent faces in $X(i)$, since the faces always share a common intersection of $i$ elements. Instead, our definition asserts that traversing through a common $(i+1)$-face is similar to traversing through a common $(i-1)$-face.

We show that this new definition is equivalent to the aforementioned definition of two-sided link expanders for constant dimension $d$, thus giving these high-dimensional expanders a new characterization.
\begin{theorem}[Equivalence between high-dimensional expander definitions] \label{thm:equivalence-main}
  Let $X$ be a $d$-dimensional simplicial complex.
  \begin{enumerate}
    \item If $X$ is a $\gamma$-two-sided link expander then X is a $\gamma$-HDX according to the definition we give.
    \item If $X$ is a $\gamma$-HDX then X is a $3(d+1)\gamma$-two-sided link expander.
  \end{enumerate}
\end{theorem}

\smallskip

Through this characterization of high-dimensional expansion, we decompose real-valued functions $f \colon X(i) \to \RR$ in an approximately orthogonal decomposition that respects the upper and lower random walk operators. We also give an example that the $O(d)$ factor in the second item of the equivalence is tight in \cref{sec:tightness-of-equivalnce-theorem}. We stress that the second direction of this equivalence is non-trivial only in the regime where $\gamma < \frac{1}{3(d+1)}$ (so when the dimension grows our definition becomes weaker).

\subsection[Decomposition of functions on X(i)]{Decomposition of functions on $X(i)$}

We begin by recalling the \emph{classical} decomposition of functions over the Boolean hypercube. Every function on the Boolean cube $\{0,1\}^n$ has a unique representation as a multilinear polynomial. In the case of the Boolean hypercube, it is convenient to view the domain as $\{1,-1\}^n$, in which case the above representation gives the \emph{Fourier expansion} of the function. The multilinear monomials can be partitioned into ``levels'' according to their degree, and this corresponds to an orthogonal decomposition of a function into a sum of its homogeneous parts, $f = \sum_{i=0}^{\deg f} f^{=i}$, a decomposition which is a basic concept in Boolean function analysis.

These concepts have known counterparts for the \emph{complete complex}, which consists of all subsets of $[n]$ of size at most $d+1$, where $d+1 \leq n/2$. The facets (top-level faces) of this complex comprise the \emph{slice} (as it is known to computer scientists) or the \emph{Johnson scheme} (as it is known to coding theorists), whose spectral theory has been elucidated by Dunkl~\cite{Dunkl1976}. For $|t| \leq d+1$, let $y_t(s) = 1$ if $t \subseteq s$ and $y_t(s) = 0$ otherwise (these are the analogs of monomials). Every function on the complete complex has a unique representation as a linear combination of monomials $\sum_t \tilde{f}(t) y_t$ (of various degrees) where the coefficients $\tilde{f}(t)$ satisfy the following \emph{harmonicity} condition:\footnote{Ryan O'Donnell has suggested the name \emph{zero-flux}, since harmonicity usually refers to vanishing of the Laplacian.} for all $i \leq d$ and all $t \in X(i)$,
\[
 \sum_{a \in [n] \setminus t} \tilde{f}(t \cup \{a\}) = 0.
\]
(If we identify $y_t$ with the product $\prod_{i \in t} x_i$ of ``variables'' $x_i$, then harmonicity of a multilinear polynomial $P$ translates to the condition $\sum_{i=1}^n \frac{\partial P}{\partial x_i} = 0$.) As in the case of the Boolean cube, this unique representation allows us to orthogonally decompose a function into its homogeneous parts (corresponding to the contribution of monomials $y_t$ with fixed $|t|$), which plays the same essential part in the complete complex as its counterpart does in the Boolean cube. Moreover, this unique representation allows extending a function from the ``slice'' to the Boolean cube (which can be viewed as a superset of the ``slice''), thus implying further results such as an invariance principle~\cite{FilmusKMW2018,FilmusM2019}.

\medskip

We generalize these concepts for complexes satisfying a technical condition we call \emph{properness}, which is satisfied when the Markov operators of the upper random walks $DU$ have full rank, or equivalently that $Ker U_i = \{0\}$ for all $i$). This holds for both the complete complex and high-dimensional expanders. We show that the results on unique decomposition for the complete complex hold for arbitrary proper complexes (\Cref{thm:decomposition}) with a generalized definition of harmonicity which incorporates the distributions $\Dist{i}$. In contrast to the case of the complete complex (and the Boolean cube), in the case of high-dimensional expanders the homogeneous parts are only approximately orthogonal.

The homogeneous components in our decomposition are ``approximate eigenfunctions'' of the Markov operators defined above, and this allows us to derive an approximate identity relating the total influence (defined through the random walks) to the norms of the components in our decomposition, in complete analogy to the same identity in the Boolean cube (expressing the total influence in terms of the Fourier expansion).
\begin{theorem}[Decomposition theorem for functions on HDX] \label{thm:decomposition-main}
Let $X$ be a proper $d$-dimensional simplicial complex.
Every function $f\colon X(\ell)\to\R$, for $\ell \leq d$, can be written uniquely as $f = f_{-1} + \cdots + f_\ell$ such that:
\begin{itemize}
\item $f_i$ is a linear combination of the functions $y_s(t) = 1_{[t \supseteq s]}$ for $s \in X(i)$, i.e.\ $y_s(t) = 1$ when $t \supseteq s$.
\item Interpreted as a function on $X(i)$, $f_i$ lies in the kernel of the Markov operator of the lower random walk $U_{i-1} D_{i}$.
\end{itemize}
If $X$ is furthermore a $\gamma$-high-dimensional expander, then the above decomposition is an almost orthogonal decomposition in the following sense:
\begin{itemize}
\item For $i \neq j$, $|\langle f_i, f_j \rangle| \approx 0$.
\item $\|f\|^2 \approx \|f_{-1}\|^2 + \cdots + \|f_\ell\|^2$.
\item If $\ell < k$ then $D_{\ell+1}U_{\ell} f_i \approx (1-\frac{i+1}{\ell+2}) f_i$, and in particular $\langle DUf,f \rangle \approx \sum_{i=-1}^\ell (1-\frac{i+1}{\ell+2}) \|f_i\|^2$.
\end{itemize}
(For an exact statement in terms of the dependence of error on $\gamma$ and $d$, see \cref{thm:approximate-orthogonality-hdx}).
\end{theorem}
Instead of requiring \(f_i\) to lie in the kernel of the walk \(UD\), they take \(f_i\) to be in the projection of the orthogonal complement of \(\operatorname{Span}\set{f_{-1},...,f_{i-1}}\). As a consequence, their decomposition is perfectly orthogonal and not just approximately orthogonal. On the flip side, the  components $f_i$ live in slightly less-nicely-defined spaces.
\prahladhuvacha{Alternate phrasing}
A similar decomposition theorem was, concurrently and independently, proved by Kaufman and Oppenheim~\cite[Theorem~1.3]{KaufmanO2020}, for one-sided spectral high-dimensional expanders. Our near-orthogonal decomposition holds only for two-sided high-dimensional expanders, while they prove an interesting norm-decomposition instead of an orthogonal decomposition, writing the norm of $f$ as the approximate sum of norms of projections of $f$ to the spaces of functions on $X(j)$. This is interesting especially in light of the fact that one-sided expanders aren't necessarily proper (for example, the complete \((d+1)\)-partite simplicial complex is not proper). In particular, there is no known way to write $f$ itself as a sum of components in analogy to \Cref{thm:decomposition-main}.

Subsequent to the earlier version of our result~\cite{DiksteinDFH2018}, Kaufman and Oppenheim~\cite[Theorem~1.5]{KaufmanO2020} extended their decomposition theorem to two-sided high-dimensional expanders. This decomposition is similar to ours but not identical. Similar to our near-orthogonal decomposition, they give a decomposition of $f$ into orthogonal components related to the spaces of functions on $X(i)$, and satisfying a similar ``nearly eigenvector'' equality for the upper walk operator. Instead of requiring \(f_i\) to lie in the kernel of the walk \(UD\), they take \(f_i\) to be in the projection of the orthogonal complement of \(\operatorname{Span}\set{f_{-1},\dots,f_{i-1}}\). As a consequence, their decomposition is perfectly orthogonal and not just approximately orthogonal. On the flip side, the components $f_i$ live in slightly less-nicely-defined spaces.

\prahladhuvacha{End of alternate phrasing}

Subsequent to our work, Alev, Jeronimo and Tulsiani~\cite{AlevJT2019} used our techniques to analyze more general random walks, which they call \emph{swap walks}. The same walks were analyzed independetly by different techniques in the work of Dikstein and Dinur~\cite{DiksteinD2019} under the name \emph{complement walks}.

\subsection{Decomposition of functions on posets} \label{sec:intro-posets}

The decomposition we suggest in this paper holds for the more general setting of graded partially ordered sets (posets): A finite \emph{graded poset} $(X,\leq, \rho)$ is a poset $(X,\leq)$ equipped with a \emph{rank} function $\rho\colon X \to \{-1\}\cup \mathbb{N}$ that respects the order, i.e.\ if $x \leq y$ then $\rho(x) \leq \rho(y)$. Additionally, if $y$ is minimal with respect to elements that are greater than $x$ (i.e. $y$ \emph{covers} $x$), then $\rho(y) = \rho(x) + 1$. Denoting $X(i) = \rho^{-1}(i)$ , we can partition the poset as follows:
\[ X = X(-1) \cup X(0) \cup \dots \cup X(d). \]
We consider graded posets with a unique minimum element $\emptyset \in X(-1)$.

Every simplicial complex is a graded poset. Another notable example is the \emph{Grassmann poset} $\Gr_q(n,d)$ which consists of all subspaces of $F_q^n$ of dimension at most $d+1$. The order is the containment relation, and the rank is the dimension minus one, $\rho(W) = \dim(W) - 1$. The Grassmann poset was recently studied in the context of proving the 2-to-1 games conjecture~\cite{KhotMS2017, DinurKKMS2018-grassman, DinurKKMS2018-2to1, KhotMS2018}, where a decomposition of functions of the Grassmann poset was useful. Such a decomposition is a special case of the general decomposition theorem in this paper.

Towards our goal of decomposing functions on graded posets, we generalize the notion of random walks on $X(i)$ as follows: A \emph{measured poset} is a graded poset with a sequence of measures $\Distv = (\Dist{-1},\ldots,\Dist{d})$ on the different levels $X(i)$, that allow us to define operators $U_i,D_{i+1}$ similar to the simplicial complex case (for a formal definition see \cref{sec:eposet}). The \emph{upper random walk} defined by the composition $D_{i+1}U_i$ is the walk where we choose two consecutive $t_1, t_2 \in X(i)$ by choosing $s \in X(i+1)$ and then $t_1, t_2 \leq s$ independently. The \emph{lower random walk} $U_{i-1}D_i$ is the walk where we choose two consecutive $t_1, t_2 \in X(i)$ by choosing $r \in X(i-1)$ and then $t_1, t_2 \geq r$ independently.

Stanley studied a special case of a measured poset that is called a \emph{sequentially differential poset}~\cite{Stanley1988}. This is a poset where
\begin{equation}\label{eq:SD-intro} D_{i+1}U_i - r_iI - \delta_i U_{i-1}D_i = 0, \end{equation}
for all $0 \leq i \leq d$ and some constants $r_i, \delta_i \in \RR$. There are many interesting examples of sequentially differential posets, such as the Grassmann poset and the complete complex. \cref{def:intro-HDX} of a high-dimensional expander resembles an approximate version of this equation:
in a simplicial complex, one may check that the non-lazy version is $\frac{i+1}{i+2} M^i_+ = D_{i+1}U_i - \frac{1}{i+2}I$. Thus
\[ \norm{M^+_i - U_{i-1}D_i} \leq \gamma \]
is equivalent to
\[ \norm{D_{i+1}U_i -\frac{1}{i+2}I - \frac{i+1}{i+2} U_{i-1}D_i} \leq \frac{i+1}{i+2}\gamma, \]
which suggests a relaxation of \eqref{eq:SD-intro} to an \emph{expanding poset} (eposet).

\begin{definition}[Expanding Poset (eposet)] \label{def:intro-eposet}
Let $\vec{r}, \vec{\delta} \in \RR_{\geq 0}^{k}$, and let $\gamma < 1$. We say $X$ is an \emph{$(\vec{r},\vec{\delta},\gamma)$-expanding poset} (or \emph{$(\vec{r},\vec{\delta},\gamma)$-eposet}) if for all $0 \leq i \leq k-1$:
\begin{equation}\label{eq:intro-asd}
\norm{D_{i+1} U_i - r_i I - \delta_i U_{i-1} D_i} \leq \gamma.
\end{equation}
\end{definition}

As we can see, $\gamma$-HDX is also an $(\vec{r},\vec{\delta},\gamma)$-eposet, for $r_i = \frac{1}{i+2},\delta_i = \frac{i+1}{i+2}$. In \cref{lem:eposet-fix-parameters} we prove that the converse is also true: every simplicial complex that is an $(\vec{r}, \vec{\delta}, \gamma)$-eposet is an $O(\gamma)$-HDX, under the assumption that the probability $\Pr_{t_1, t_2 \sim U_{i-1}D_i}[t_1=t_2]$ is small.

It turns out that eposets are the correct setup to generalize our decomposition for functions on simplicial complexes: in all eposets we can uniquely decompose functions $f \colon X(i) \to \RR$ to
\[ f = \sum_{j=-1}^i f^{=j}, \]
where the functions $f^{=j}$ are ``approximate eigenvectors'' of $D_{i+1}U_i$. Furthermore, this decomposition is ``approximately orthogonal''. Fixing $i$, the error in both approximations is $O(\gamma)$.

\subsection{An FKN theorem}
Returning to simplicial complexes, as a demonstration of the power of this setup, we generalize the fundamental result of Friedgut, Kalai, and Naor~\cite{FriedgutKN2002} on Boolean functions almost of degree~1.
We view this as a first step toward developing a full-fledged theory of Boolean functions on high-dimensional expanders.

An easy exercise shows that a Boolean degree~1 function on the Boolean cube is a \emph{dictator}, that is, depends on at most one coordinate; we call this the \emph{degree one theorem} (the easy case of the FKN Theorem with zero-error). The FKN theorem, which is the robust version of this degree one theorem,  states that a Boolean function on the Boolean cube which is \emph{close} to a degree~1 function is in fact close to a dictator, where closeness is measured in $L_2$.

The degree one theorem holds for the complete complex as well. The third author~\cite{Filmus2016-fkn} has extended the FKN theorem to the complete complex. Surprisingly, the class of approximating functions has to be extended beyond just dictators.

We prove a degree one theorem for arbitrary proper complexes, and an FKN theorem for high-dimensional expanders. In contrast to the complete complex, Boolean degree~1 functions on arbitrary complexes correspond to \emph{independent sets} rather than just single points, and this makes the proof of the degree one theorem non-trivial.
\begin{definition}[1-skeleton]
The \emph{$1$-skeleton} of a simplicial complex $X$ is the graph whose vertices are $X(0)$, the $0$-faces of the complex, and whose edges are $X(1)$, the $1$-faces of the complex.
\end{definition}
\begin{claim}[Degree one theorem on simplicial complexes] \label{thm:exact-fkn-main}
Suppose that $X$ is a proper $d$-dimensional simplicial complex, for $d \geq 2$, whose $1$-skeleton is connected. A function $f \in C^d$ is a Boolean degree~$1$ function if and only if there exists an independent set $I$ (in the $1$-skeleton of $X$) such that $f$ is the indicator of intersecting $I$ or of not intersecting $I$.
\end{claim}

Our proof of the FKN theorem for high-dimensional expanders is very different from existing proofs. It follows the same general plan as recent work on the biased Kindler--Safra theorem~\cite{DinurFH2019}. The idea is to view a high-dimensional expander as a convex combination of small sub-complexes, each of which is isomorphic to the complete $k$-dimensional complex on $O(k)$ vertices. We can then apply the known FKN theorem separately on each of these, and deduce that our function is approximately well-structured on each sub-complex. Finally, we apply the agreement theorem of Dinur and Kaufman~\cite{DinurK2017} to show that the same holds on a global level.

\begin{theorem}[FKN theorem on HDX (informal)]\label{thm:fkn-main}
Let $X$ be a $d$-dimensional $\gamma$-HDX. If $F\colon X(d) \to \{0,1\}$ is $\epsilon$-close (in $L_2^2$) to a degree~$1$ function then there exists a degree~$1$ function $g$ on $X(d)$ such that $\Pr[F \neq g] = O_{\gamma,d}(\epsilon)$.
\end{theorem}
\subsection*{Paper organization}
We describe our general setup in \cref{sec:basic-setup}. We describe the property of properness and its implications --- a unique representation theorem and decomposition of functions into homogeneous parts --- in \cref{sec:decomposition}.  We introduce our definition of high-dimensional expanders in \cref{sec:HD-expanders-def}. In \cref{sec:HD-expanders} we show equivalence between our definition and the earlier one of two-sided link expanders.

In \cref{sec:eposet} we define expanding posets, and describe our decomposition of functions on posets. We prove almost orthogonality of the decomposition (\cref{thm:decomposition}) in this more general setting.
The full decomposition theorem for the (interesting) special case of simplicial complexes is explicitly stated in \cref{thm:approximate-orthogonality-hdx}.

We prove our degree one theorem in \cref{sec:exact-fkn}, and our FKN theorem in \cref{sec:fkn}.

\cref{thm:decomposition-main} is a combination of \cref{thm:decomposition} (first two items) and \cref{thm:approximate-orthogonality-hdx} (other three items).
\cref{thm:equivalence-main} is a restatement of \cref{thm:equivalence}.
\cref{thm:exact-fkn-main} is a restatement of \cref{thm:hard-fkn}. \cref{thm:fkn-main} is a restatement of \cref{thm:fkn-hdx}.

\section{Basic setup} \label{sec:basic-setup}

A $d$-dimensional simplicial complex $X$ is a non-empty collection of sets of size at most $d+1$ which is closed under taking subsets. We call a set of size $i+1$ an \emph{$i$-dimensional face} (or $i$-face for short), and denote the collection of all $i$-faces by $X(i)$. A $d$-dimensional simplicial complex $X$ is \emph{pure} if every $i$-face is a subset of some $d$-face. We will only be interested in pure simplicial complexes.

Let $X$ be a pure $d$-dimensional simplicial complex. Given a probability distribution $\Dist{d}$ on its top-dimensional faces $X(d)$, we define a distribution $\Distv = (\Dist{d},\ldots,\Dist{-1})$ over sequences $s_d \supsetneq s_{d-1} \supsetneq \ldots \supsetneq s_{-1} = \emptyset$ where $s_i \in X(i)$ as follows. We sample $s_d = \{v_0,v_1,...,v_d \} \in X(d)$ according to $\Dist{d}$, and order its vertices uniformly at random $(v_0,v_1,...,v_d)$. Then we set $s_i = \{v_0,v_1,...,v_i\}$. The distribution $\Dist{i}(s_i)$ is the probability of sampling $s_i$.

Let $C^i := \{ f\colon X(i)\to\R\} $ be the space of functions on $X(i)$. It is convenient to define $X(-1):=\{\emptyset\}$, and we also let $C^{-1}:=\R$. We turn $C^i$ to an inner product space by defining $\langle f,g \rangle := \EE_{\Dist{i}}[fg]$ and the associated norm $\|f\|^2 := \EE_{\Dist{i}}[f^2]$.

For $-1 \leq i < d$, we define the {\bf Up} operator $U_i\colon C^i \to C^{i+1}$ as follows:\footnote{The Up and Down operators differ from the boundary and coboundary operators of algebraic topology, which operate on linear combinations of \emph{oriented} faces.}
\[
 U_ig(s) := \frac{1}{i+2} \sum_{x \in s} g(s \setminus \{x\}) = \E_{t \subset s}[g(t)]\;,
\]
where $t$ is obtained from $s$ by removing a random element. Note that if $s \sim \Dist{i+1}$ then $t \sim \Dist{i}$.

Similarly, we define the {\bf Down} operator $D_{i+1}\colon C^{i+1} \to C^i$ for $-1 \leq i < d$ as follows:
\[
 D_{i+1} f(t) := \frac{1}{(i+2)\cdot\Dist{i}(t)} \sum_{x \notin t\colon t \cup \{x\} \in X(i+1)} \Dist{i+1}(t \cup \{x\}) \cdot f(t \cup \{x\}) = \E_{s \supset t}[f(s)]\;,
\]
where $s$ is obtained from $t$ by conditioning the vector $\Distv$ on $\Dist{i} = t$ and taking the $(i+1)$th component.

The operators $U_i,D_{i+1}$ are adjoint to each other. Indeed, if $f \in C^{i+1}$ and $g \in C^i$ then
\begin{equation}\nonumber
\langle g,D_{i+1}f \rangle = \E_{(t,s) \sim (\Dist{i},\Dist{i+1})} [g(t) f(s)] = \langle U_ig,f \rangle\;.
\end{equation}
When the domain is understood, we will use $U,D$ instead of $U_i,D_{i+1}$. This will be especially useful when considering powers of $U,D$. For example, if $f\colon X(i) \to \R$ then
\[
 U^t f \equiv U_{i+t-1} \ldots U_{i+1} U_i f.
\]

Given a face $s \in X$, the function $y_s$ is the indicator function of containing $s$. Our definition of the Up operator guarantees the correctness of the following lemma.

\begin{lemma} \label[lemma]{lem:ys-formula}
 Let $s \in X(i)$. We can think of $y_s$ as a function in $C^j$ for all $j \geq i$. Using this convention, $U_j y_s = (1-\frac{i+1}{j+2}) y_s$.
\end{lemma}
\begin{proof}
 Direct calculation shows that
\[
 (U_j y_s)(t) = \frac{1}{j+2} \sum_{x \in t} y_s(t \setminus \{x\}) = \frac{|t|-|s|}{j+2} y_s(t)\;,
\]
 and so $U_j y_s = (1-\frac{i+1}{j+2}) y_s$.
\end{proof}

For $0 \leq i \leq k$, the space of harmonic functions on $X(i)$ is defined as
\[ H^i := \ker D_i = \{ f \in C^i : D_if=0\}\;. \]
We also define $H^{-1} := C^{-1}=\R$. We are interested in decomposing $C^k$, so let us define for each $-1 \leq i \leq k$,
\[ V^i := U^{k-i}H^i = \{ U^{k-i}f \;:\; f\in H^i \}\;.\]
We can describe $V^i$, a sub-class of functions of $C^k$, in more concrete terms.

\begin{lemma} \label[lemma]{lem:harmonic-concrete}
Every function $h \in V^i$ has a representation of the form
\[
 h = \sum_{s \in X(i)} \tilde{h}(s) y_s\;,
\]
where the coefficients $\tilde{h}(s)$ satisfy the following \emph{harmonicity} condition: for all $t \in X(i-1)$,
\[
 \sum_{s \supset t} \Dist{i}(s) \tilde{h}(s) = 0\;.
\]
Furthermore, if $U^{k-i}$ is injective on $C^i$ then the representation is unique.
\end{lemma}
\begin{proof}
 Suppose that $h \in V^i$. Then $h = U^{k-i}f$ for some $f \in H^i$, which by definition of $H^i$ and the Down operator is equivalent to the condition
 \[
  \sum_{s \supset t} \Dist{i}(s) f(s) = 0
 \]
 for all $t \in X(i-1)$. In other words, the $f(s)$'s satisfy the harmonicity condition. It is easy to check that $f = \sum_{s \in X(i)} f(s) y_s$, and so \Cref{lem:ys-formula} shows that $h = \sum_{s \in X(i)} \tilde{h}(s) y_s$, where
\[
 \tilde{h}(s) = \left(1-\frac{i+1}{k+1}\right) \cdots \left(1-\frac{i+1}{i+2}\right) f(s).
\]
Thus, $\tilde{h}(s)$ is a scaling of $f(s)$ by a non-zero constant, it follows that the coefficients $\tilde{h}(s)$ also satisfy the harmonicity condition.

Now suppose that $U^{k-i}$ is injective on $C^i$, which implies that $\dim H^i = \dim V^i$. The foregoing shows that the dimension of the space of coefficients $\tilde{h}(s)$ satisfying the harmonicity conditions is $\dim H^i$. Since $\dim H^i = \dim V^i$, this shows that the representation is unique.
\end{proof}

\section{Decomposition of the space \texorpdfstring{$C^k$}{Ck} and a convenient basis} \label{sec:decomposition}

Our decomposition theorem relies on a crucial property of simplicial complexes, \emph{properness}.

\begin{definition} \label[definition]{def:proper}
A $k$-dimensional simplicial complex is \emph{proper} if $D_{i+1}U_{i}>0$ (i.e.\ $D_{i+1}U_i$ is positive definite) for all $i\leq k-1$. Equivalently, if it is proper and $\ker U_i$ is trivial for $-1 \leq i \leq k-1$.
\end{definition}

We remark that since $DU$ is PSD, $\ker U= 0$ is equivalent to $DU>0$. This is because for any $x\in \ker DU$, we would have $0 = \langle x, DUx \rangle = \norm{U x}^2$, implying that $x=0$.

The complete $k$-dimensional complex on $n$ points is proper iff $k+1 \leq \frac{n+1}{2}$.
A pure one-dimensional simplicial complex (i.e., a graph) is proper iff it is not bipartite. Unfortunately, we are not aware of a similar characterization for higher dimensions. However, in \cref{sec:HD-expanders} we show that high-dimensional expanders are proper.

We can now state our decomposition theorem.

\begin{theorem} \label{thm:decomposition}
If $X$ is a proper $k$-dimensional simplicial complex then we have the following decomposition of $C^k$:
\[ C^k = V^k \oplus V^{k-1} \oplus \dots \oplus V^{-1}\;. \] In other words, for every function $f\in C^k$ there is a unique choice of $h_i \in H^i$ such that the functions $f_i = U^{k-i}h_i$ satisfy $f = f_{-1} + f_{0} + \dots + f_k$.
\end{theorem}
\begin{proof}
We first prove by induction on $\ell$ that every function $f \in C^\ell$ has a representation $f = \sum_{i=-1}^\ell U^{\ell-i} h_i$, where $h_i \in H^i$. This trivially holds when $\ell = -1$. Suppose now that the claim holds for some $\ell<k$, and let $f \in C^{\ell+1}$. Since $D^{\ell+1}\colon C^{\ell+1} \to C^\ell$ is a linear operator, we can decompose $C^{\ell+1}$ to $\ker D_{\ell+1} \oplus (\ker D_{\ell+1})^\bot$. It is well known that $(\ker D_{\ell+1})^\bot = \im D_{\ell+1}^\ast$ so we have $C^{\ell+1} = \ker D_{\ell+1} \oplus \im D_{\ell+1}^\ast = \ker D_{\ell+1} \oplus \im U_\ell$, and therefore we can write $f = h_{\ell+1} + Ug$, where $h_{\ell+1} \in H^{\ell+1}$ and $g \in C^\ell$. Applying induction, we get that $g = \sum_{i=-1}^\ell U^{\ell-i} h_i$, where $h_i \in H^i$. Substituting this in $f = h_{\ell+1} + Ug$ completes the proof.

It remains to show that the representation is unique. Since $\ker U_{i-1} = \ker D_i^*$ is trivial, $\dim H^i = \dim C^i - \dim C^{i-1}$ for $i \geq 0$. This shows that $\sum_{i=-1}^k \dim H^i = \dim C^k$. Therefore the operator $\varphi\colon H^{-1} \times \cdots \times H^k \to C^k$ given by $\varphi(h_{-1},\ldots,h_k) = \sum_{i=-1}^k U^{k-i} h_i$ is not only surjective but also injective. In other words, the representation of $f$ is unique.
\end{proof}

\begin{corollary} \label[corollary]{cor:decomposition}
If $X$ is a proper $k$-dimensional simplicial complex then every function $f \in C^k$ has a unique representation of the form
\[
 f = \sum_{s \in X} \tilde{f}(s) y_s\;,
\]
where the coefficients $\tilde{f}(s)$ satisfy the following \emph{harmonicity} conditions: for all $0 \leq i \leq k$ and all $t \in X(i-1)$:
\[
 \sum_{\substack{s \in X(i) \\ s \supset t}} \Dist{i}(s) \tilde{f}(s) = 0\;.
\]
\end{corollary}
\begin{proof}
Follows directly from \cref{lem:harmonic-concrete}.
\end{proof}

We can now define the degree of a function.
\begin{definition} \label[definition]{def:degree}
The \emph{degree} of a function $f$ is the maximal cardinality of a face $s$ such that $\tilde{f}(s) \neq 0$ in the unique decomposition given by \cref{cor:decomposition}.
\end{definition}

Thus a function has degree~$d$ if its decomposition only involves faces whose dimension is less than~$d$. The following lemma shows that the functions $y_s$, for all $(d-1)$-dimensional faces~$s$, form a basis for the space of all functions of degree at most~$d$.

\begin{lemma} \label[lemma]{lem:degree-rep}
If $X$ is a proper $k$-dimensional simplicial complex then the space of functions on $X(k)$ of degree at most $d+1$ has the functions $\{ y_s : s \in X(d) \}$ as a basis.
\end{lemma}
\begin{proof}
The space of functions on $X(k)$ of degree at most $d+1$ is spanned, by definition, by the functions $y_t$ for $t\in X(-1)\cup X(0)\cup\cdots\cup X(d)$. This space has dimension $\sum_{i=-1}^d \dim H^i$. Since $X$ is proper, $\dim H^i = \dim C^i - \dim C^{i-1}$ for $i > 0$, and so $\sum_{i=1}^d \dim H^i = \dim C^d = |X(d)|$.

Given the above, in order to complete the proof, it suffices to show that for every $i \leq d$ and $t \in X(i)$, the function $y_t$ can be written as a linear combination of $y_s$ for $s \in X(d)$. This shows that $\{ y_s : s \in X(d) \}$ spans the space of functions of degree at most~$d+1$. Since this set contains $|X(d)|$ functions, it forms a basis.

Recall that $y_t(r) = 1_{r \supseteq t}$, where $r \in X(k)$. If $r$ contains $t$ then it contains exactly $\binom{k+1-|t|}{d+1-|t|}$ many $d$-faces containing $r$, and so
\[
 y_t = \frac{1}{\binom{k+1-|t|}{d+1-|t|}} \sum_{\substack{s \supseteq t \\ s \in X(d)}} y_s.
\]
This completes the proof.
\end{proof}

We call $f_i$ the ``level $i$'' part of $f$, and denote the weight of $f$ above level $i$ by
\[ wt_{>i}(f) := \sum_{j>i} \|f_j\|_2^2.\]
We also define $f_{\leq i} = f_{-1} + \cdots + f_i$ and $f_{>i} = f - f_{\leq i}$.

\section{How to define high-dimensional expansion?} \label[section]{sec:HD-expanders-def}

In this section we define a class of simplicial complexes which we call $\gamma$-high-dimensional expanders (or $\gamma$-HDXs). We later show that these simplicial complexes coincide with the high-dimensional expanders defined by Dinur and Kaufman~\cite{DinurK2017} via spectral expansion of the links. In addition, we show the decomposition in \cref{sec:decomposition} is almost orthogonal for $\gamma$-HDXs. We define $\gamma$-HDXs through relations between random walks in different dimensions. It is easy to already state the definition using the $U,D$ operators: a $k$-dimensional simplicial complex is said to be a $\gamma$-HDX if for all levels $0 \leq j \leq k-1$,
\begin{equation} \label{eq:SD-def-HDX}
\left\|{\frac{j+2}{j+1}\left(DU - \frac{1}{j+2} I\right) - UD}\right\| \leq \gamma.
\end{equation}
We turn to explain the meaning of \eqref{eq:SD-def-HDX} being small by discussing these random walks.\footnote{$UD$ and $DU$ are called high-dimensional Laplacians in some other works~\cite{KaufmanO2020}.}

\medskip

The operators $U$ and $D$ induce random walks on the $j$th level $X(j)$ of the simplicial complex. Recall that our simplicial complexes come with distributions $\Dist{j}$ on the $j$-faces.
\begin{definition}[The upper random walk $DU$]
Given  $t \in X(j)$, we choose the next set $t' \in X(j)$ as follows:
\begin{itemize}
\item Choose $s \sim \Dist{j+1}$ conditioned on $t \subset s$.
\item Choose uniformly at random $t' \in X(j)$ such that $t' \subset s$.
\end{itemize}
\end{definition}

\begin{definition}[The lower random walk $UD$]
Given  $t \in X(j)$, we choose the next set $t' \in X(j)$ as follows:
\begin{itemize}
\item Choose uniformly at random $r \in X(j-1)$ such that $r \subset t$.
\item Choose $t' \sim \Dist{j}$ conditioned on $r \subset t'$.
\end{itemize}
\end{definition}
It is easy to see that the stationary distribution for both these processes is $\Dist{j}$. However, these random walks are not necessarily the same. For example, if $j=0$, we consider the graph $(X(0),X(1))$. The upper walk is the $\frac{1}{2}$-lazy version of the usual adjacency random walk in a graph. The lower random walk is simply choosing two vertices independently, according the distribution $\Dist{0}$.
In both walks, the first step and the third step are independent given the second step. In fact, we can view the upper walk (resp.\ lower walk) as choosing a set $s \in X(j+1)$ (resp.\ $r \in X(j-1)$), and then choosing independently two sets $t,t' \in X(j)$ given that they are contained in $s$ (resp.\ given that they contain $r$).

One property of a random walk is its laziness:
\begin{definition}[Laziness and non-lazy component]\label{def:nonlazy}
Let $M$ be a random walk. The \emph{laziness} of $M$ is
\[ \ell z(M) = \Pr_{(x,y) \sim M}[x=y]. \]
We say that an operator is \emph{non-lazy} if $\ell z(M) = 0$.

The non-lazy component $M^+$ of a walk $M$ is given by
\[  M^+ (x,y) = \begin{cases}
    \frac{M(x,y)}{\sum_{y' \neq x} M(x,y')} & \text{ if } y \neq x,\\
    0 & \text{ if } y = 0.
  \end{cases}.
  \]
\end{definition}

It is easy to see that both walks have some laziness. In the upper walk, the laziness is $\frac{1}{j+2}$. We can decompose $DU$ as
\begin{equation}\label{eq:lazy-non-lazy-decomposition}
DU = \frac{1}{j+2} I + \frac{j+1}{j+2}M^+_j,
\end{equation}
where $M^+_j$ is the non-lazy version of $DU$, i.e.\ the operator representing the walk when conditioning on $t' \ne t$. The laziness of the lower version depends on the simplicial complex itself, thus it doesn't admit a simple decomposition in the general case.

\eqref{eq:lazy-non-lazy-decomposition} can be written as
\[ M^+_j = \frac{j+2}{j+1} \left(DU - \frac{1}{j+2}I\right). \]

A $\gamma$-HDX is a simplicial complex in which the non-lazy upper walk is similar to the lower walk. Thus an equivalent way to state \eqref{eq:SD-def-HDX} is as follows.

\begin{definition}[High-dimensional expander] \label{def:HDX}
Let $X$ be a simplicial complex, and let $\gamma < 1$. We say that $X$ is a \emph{$\gamma$-HDX} if for all $0 \leq j \leq k-1$,
\begin{equation} \label{eq:HDX}
\|{M^+_j - UD}\| \leq \gamma.
\end{equation}
\end{definition}

This definition nicely generalizes spectral expansion in graphs, since if $X$ is a graph, $\|{M^+_j - UD}\| $ is the second largest eigenvalue (in absolute value) of the normalized adjacency random walk. In \cref{sec:HD-expanders} we show that this definition is equivalent to the definition of high-dimensional two-sided local spectral expanders that was extensively studied~\cite{DinurK2017, Oppenheim2018}.

If $\gamma < \frac{1}{k+1}$ then any $\gamma$-HDX is proper, as shown by the following lemma.

\begin{lemma} \label{clm:HDXs-are-proper}
Let $X$ be a $k$-dimensional $\gamma$-HDX, for $\gamma < \frac{1}{k+1}$. Then $X$ is proper.
\end{lemma}

\begin{proof}
To prove this, we directly calculate $ \iprod{U_j f, U_j f}$ and show that it is positive when $f \ne 0$:
\begin{align}
\iprod{Uf, Uf} &= \iprod{DUf,f} = \frac{1}{j+2}\iprod{f,f} + \frac{j+1}{j+2}\iprod{M^+_j f,f} \notag \\ &=
\label{eq:HDXs-are-proper}
 \frac{1}{j+2}\iprod{f,f} + \frac{j+1}{j+2}\iprod{(M^+_j -UD + UD)f,f}.
\end{align}

From Cauchy--Schwarz,
\[ \abs{\iprod{(M^+_j -UD) f, f}} \leq \norm{(M^+_j -UD) f} \norm{f}, \]
and since $X$ is a $\gamma$-HDX,
\[ \norm{(M^+_j -UD) f} \leq \gamma \norm{f}.\]
Plugging this in \eqref{eq:HDXs-are-proper}, we get
\[  \frac{1}{j+2}\iprod{f,f} + \frac{j+1}{j+2}\iprod{(M^+_j -UD + UD)f,f} \geq \left(\frac{1}{j+2}-\frac{j+1}{j+2}\gamma\right)\iprod{f,f}  + \iprod{UD f,f}. \]
The last part of the sum is non-negative: $\iprod{UDf, f} = \iprod{Df,Df} \geq 0$. Therefore, if $\gamma <\frac{1}{k+1} \leq \frac{1}{j+2}$ then
\[ \left(\frac{1}{j+2}-\frac{j+1}{j+2}\gamma\right)\iprod{f,f}  + \iprod{DU f,f} \geq \left(\frac{1}{j+2}-\frac{j+1}{j+2}\gamma\right)\iprod{f,f}  > 0. \]
Hence $\iprod{U_j f, U_j f} > 0$.
\end{proof}

\subsection{Almost orthogonality of the decomposition in HDXs} \label{sec:laplacian}

In \cref{sec:eposet} we prove that the decomposition in \cref{thm:decomposition} is ``almost orthogonal''. We summarize our results below:

\begin{theorem} \label[theorem]{thm:approximate-orthogonality-hdx}
Let $X$ be a $k$-dimensional $\gamma$-HDX, where $\gamma<\frac{1}{k+1}$.
For every function $f$ on $C^\ell$ for $\ell \leq k$, the decomposition $f = f_{-1} + \cdots + f_\ell$ of \cref{thm:decomposition} satisfies the following properties:
\begin{itemize}
\item For $i \neq j$, $|\langle f_i, f_j \rangle| = O(\gamma) \|f_i\| \|f_j\|$.
\item $\|f\|^2 = (1 \pm O(\gamma))(\|f_{-1}\|^2 + \cdots + \|f_\ell\|^2)$, and for all $i$, $\|f\|^2 = (1 \pm O(\gamma))(\|f_{\leq i}\|^2 + \|f_{>i}\|^2)$.
\item $f_i$ are approximate eigenvectors with eigenvalues $\lambda_i = 1-\frac{i+1}{\ell+2}$ in the sense that $\|DU f_i - (1-\frac{i+1}{\ell+2})f_i\| = O(\gamma)\norm{f_i}$.
\item If $\ell < k$ then $\langle DUf,f \rangle = (1 \pm O(\gamma)) \sum_{i=-1}^\ell \lambda_i \|f_i\|^2$.
\end{itemize}
The hidden constants in the $O$ notations depend on $k$ but not on the size of $X$.
\end{theorem}

This result is analogous to a similar result \cite[Theorem 5.2]{KaufmanO2020} due to Kaufman and Oppenheim, in which a similar decomposition is obtained. However, whereas our decomposition is to functions $f_{-1},\ldots,f_\ell$ in $C^\ell$, the decomposition of Kaufman and Oppenheim~\cite{KaufmanO2020} is to functions $h_{-1},\ldots,h_\ell$, which live in different spaces. We further expand on the matter in \Cref{sec:decomposition-ko}.

The exact constants in the $O$ notation of \cref{thm:approximate-orthogonality-hdx} were not calculated precisely, but the proof in \cref{sec:eposet} gives constants exponential in $k$. It is not clear whether this is tight.

\section{High-dimensional expanders are two-sided link expanders} \label{sec:HD-expanders}
In \cref{sec:HD-expanders-def} we defined $\gamma$-HDXs, see \cref{def:HDX}. Earlier works,~\cite{EvraK2016,DinurK2017,KaufmanO2020} for example, gave a different definition of high-dimensional expanders --- \emph{two-sided link expanders} --- based on the local link structure. We recall this other definition and prove that the two are equivalent.

\medskip

\begin{definition}[Link]
Let $X$ be a $d$-dimensional complex with an associated probability distribution $\Dist{d}$ on $X(d)$, which induces probability distributions on $X(-1),\ldots,X(d-1)$. For every $i$-dimensional face $s \in X(i)$ for $i < d-1$, the \emph{link} of $s$, denoted $X_s$, is the simplicial complex:
\[X_s = \{r \setminus s : r \in X, r \supset s\}. \]
We associate $X_s$ with the weights $\Distv^s$ such that
\[ \Dist{j}^s(t) := \Pr_{r \sim \Dist{i+j+1}}[r=s \cup t | r \supset s] = \frac{\Dist{}(s \cup t)}{\Dist{}(s) \binom{\abs{s \cup t}}{\abs{s}}}.\]
\end{definition}

\begin{definition}[Underlying graph]
Let $ i < d-1$. Given $s \in X(i)$, the \emph{underlying graph} $G_s$ is the weighted graph consisting of the first two levels of the link of $s$. In other words, $G_s = (V,E)$, where
\begin{itemize}
\item $V = X_s(0) = \{x \notin s: s \cup \{x\} \in X(i+1) \}. $
\item $E = X_s(1) = \{ \{x,y\} : s \cup \{x,y\} \in X(i+2) \}. $
\end{itemize}
The weights on the edges are given by
\[ w_s(\{x,y\}) = \Pr_{r \sim \Dist{i+2}}[r = s \cup \{x,y\} | r \supset s] = \frac{\Dist{}(s \cup \{x,y\})}{\Dist{}(s) \binom{\abs{s}+2}{\abs{s}}}. \]
\end{definition}
We can also consider directed edges, by choosing a random orientation:
\[ w_s(x,y) = \frac{1}{2} w_s(\{x,y\}). \]

We define the weight of a vertex $x$ to
\[
w_s(x) := \Dist{0}^s(x) = \Pr_{r \sim \Dist{i+1}}[r=s \cup \{x\} | r \supset s] = \frac{\Dist{}(s \cup \{x\})}{\Dist{}(s) (|s|+1)}.
\]
We define an inner product for functions on vertices along the lines of \cref{sec:basic-setup}:
\[
 \langle f,g \rangle := \EE_{x \sim w_s}[f(x) g(x)].
\]

We denote by $A_s$ the adjacency operator of the \emph{non-lazy upper-walk} on $X_s(0)$, given by
\[ A_sf(x) = \EE_{y \sim w_s}[f(y) | \{x,y\} \in E]. \]
The corresponding quadratic form is
\[
 \langle f,A_s g \rangle = \EE_{(x,y) \sim w_s} [f(x) g(y)].
\]
By definition, $A_s$ fixes constant functions, and is a Markov operator. It is self-adjoint with respect to the inner product above. Thus $A_s$ has eigenvalues $\lambda_1 = 1 \geq \lambda_2 \geq \ldots \geq \lambda_m$, where $m$ is the number of vertices. We define $\lambda(A_s) = \max(|\lambda_2|,|\lambda_m|)$. Orthogonality of eigenspaces guarantees that
\begin{equation} \label{eq:link-expander-key-property}
| \langle f,A_sg \rangle - \EE[f] \EE[g] | \leq \lambda(A_s) \|f\| \|g\|.
\end{equation}

\begin{definition}[Two-sided link expander] \label{def:link-HDX}
  Let $X$ be a simplicial complex, and let $\gamma < 1$ be some constant. We say that $X$ is a \emph{$\gamma$-two-sided link expander} (called $\gamma$-HD expander in some works~\cite{DinurK2017}) if every link $X_s$ of $X$ satisfies $\lambda(A_s) \leq \gamma$.
\end{definition}
Dinur and Kaufman~\cite{DinurK2017} proved that such expanders do exist, based on a result of Lubotzky, Samuels and Vishne~\cite{LubotzkySV2005-exphdx}.

\begin{theorem}[{\cite[Lemma 1.5]{DinurK2017}}] \label{thm:expanders-exist}
For every $\lambda > 0$ and every $d \in \mathbb{N}$ there exists an explicit infinite family of bounded degree $d$-dimensional complexes which are $\lambda$-two-sided link expanders.
\end{theorem}

We now prove that \emph{two-sided link expanders} per \cref{def:link-HDX} and \emph{high-dimensional expanders} per \cref{def:HDX} are equivalent.

\begin{theorem}[Equivalence theorem] \label{thm:equivalence}
  Let $X$ be a $d$-dimensional simplicial complex.
  \begin{enumerate}
    \item If $X$ is a $\gamma$-two-sided link expander, then $X$ is a $\gamma$-HDX.
    \item If $X$ is a $\gamma$-HDX then $X$ is a $3d\gamma$-two-sided link expander.
  \end{enumerate}
\end{theorem}

\begin{proof}
\textbf{Item 1.} Assume that $X$ is a $\gamma$-two-sided link expander. We need to show that
\[\norm{M^+_i - UD} \leq \gamma, \]
for all $i < d$, where $M^+_i$ is the non-lazy upper walk.
Let $f$ be a function on $X(i)$, where $i < d$. We have
\[
 \langle M^+_i f, f \rangle =
 \EE_{t \sim \Dist{i+1}} \EE_{x \neq y \in t} [f(t \setminus \{x\}) f(t \setminus\{y\}].
\]
Let $s = t \setminus \{x,y\}$. Since $t \sim \Dist{i+1}$ and $x \neq y \in t$ are chosen at random, we have $s \sim \Dist{i-1}$. Given such an $s$, the probability to get specific $(t,x,y)$ is exactly $w_s(x,y)$ (the factor $1/2$ accounts for the relative order of $x,y$), and so
\begin{equation}
\langle M^+_i f, f \rangle = \EE_{s \sim \Dist{i-1}} \EE_{(x,y) \sim w_s} [f(s \cup \{x\}) f(s \cup \{y\})].
\end{equation}

In other words, we have shown that
\begin{equation} \label{eq:localization-equation}
 \langle M^+_i f, f \rangle = \EE_{s \sim \Dist{i-1}}[ \langle A_s f_s, f_s \rangle],
\end{equation}
where $f_s\colon X_s(0) \to \RR$ is defined by
\[ f_s(x) = f(s \cup \{x\}).\]

We now note that
\[
 \EE_{x \sim w_s}[f(s \cup \{x\})] = (Df)(s).
\]
Therefore we have, by \eqref{eq:link-expander-key-property}, that
\begin{multline*}
\abs{\iprod{M^+_i f, f} - \iprod{UD f, f} } =  \bigl|\EE_{s \sim \Dist{i-1}} \EE_{(x,y) \sim w_s} [f(s \cup \{x\}) f(s \cup \{y\})] - (Df)(s)^2\bigr| \leq \\ \EE_{s \sim \Dist{i-1}} \bigl[\lambda(A_s) \EE_{x \sim w_s}[f(s \cup \{x\})^2]\bigr].
\end{multline*}
If $X$ is a $\gamma$-two-sided link expander then $\lambda(A_s) \leq \gamma$ for all $s$, and so
\[
 \abs{\iprod{(M^+_i - UD) f, f}} \leq \gamma \|f\|^2.
\]

\medskip

\textbf{Item 2.}
 Assume now that $X$ is a $\gamma$-HDX. Our goal is to show that for all $i < d-1$ and $r \in X(i)$,
\[ \lambda(A_r) \leq 3(i+2)\gamma.\]

Using the convention that $X(-1)$ consists of the empty set, for $i=-1$ we have $A_{\emptyset} = M^+_0$, and so $U_{-1}D_{0}$ is zero on the space perpendicular to the constant function. Thus
\[\norm{M^+_0 - UD}  = \lambda(A_\emptyset), \]
and from our assumption $\lambda(A_\emptyset) \leq \gamma$.

Now assume $1 \leq i \leq d-1$, and fix some $r \in X(i-1)$. Let $f\colon X_r(0) \to \RR$ be some eigenfunction of $A_r$, which is perpendicular to the constant function. In order to prove the theorem, we must show that
\[\abs{\frac{\iprod{A_r f, f}}{\iprod{f,f}}} \leq 3(i+1)\gamma.\]

Define a function $\tilde{f} \in C^{i}$ by
\[ \tilde{f}(s) = \begin{cases}
    f(s \setminus r) &  \text{if }r \subset s, \\
    0 & \text{otherwise}.
  \end{cases} \]
Without loss of generality, we may assume that $\|\tilde{f}\|=1$.

In order to obtain a bound on $\lambda(A_r)$, we bound $\iprod{\tilde{f},\tilde{f}}$, $\iprod{M^+_i \tilde{f}, \tilde{f}}$, and $\iprod{UD \tilde{f},\tilde{f}}$ in terms of $f$ and $A_r$.

Observe that the norms of $f$ and $\tilde{f}$ are proportional:
\begin{equation}
  \label{eq:converse-1}
  \iprod{f,f} = \frac{\iprod{\tilde{f},\tilde{f}}}{\Dist{i-1}(r)(i+1)} = \frac{1}{\Dist{i-1}(r)(i+1)}.
\end{equation}
Furthermore, from what we showed in \eqref{eq:localization-equation} we obtain that \[ \iprod{M^+_i \tilde{f}, \tilde{f}} = \EE_{r' \in X(i-1)}[\iprod{A_{r'} \tilde{f}_{r'}, \tilde{f}_{r'}}],\]
where $\tilde{f}_{r'}(x) = \tilde{f}(r' \cup \{x\}).$

Fix some $r' \ne r$. If $\tilde{f}_{r'}(x) \ne 0$ then $\tilde{f}(r' \cup \{x\}) \ne 0$. In particular, this means that $r \subset r' \cup \{x\}$. Since both $r,r'$ are contained in $r' \cup \{x\}$, this means that $r \setminus r' = \{x\}$. Thus there is at most one vertex $x \in X_{r'}(0)$ such that $\tilde{f}_{r'}(x) \ne 0$. Since $A_{r'}$ is a non-lazy operator, this implies that $\iprod{A_{r'} \tilde{f}_{r'}, \tilde{f}_{r'}} = 0$. We remain with
\begin{equation} \label{eq:converse-2}
  \iprod{M^+_i \tilde{f}, \tilde{f}} = \Dist{i-1}(r)\iprod{A_r f, f}.
\end{equation}
In other words, the upper non-lazy random walk is proportional to the local adjacency operator.

We  prove the following claim, which shows that the lower walk scales $\tilde{f}$ by a factor of at most $\frac{i+1}{i}\gamma$:
\begin{claim}\label{f-down-operator-bounding-claim}
	If $f\colon X_r(0) \to \mathbb{R}$ is perpendicular to constant functions then
  	$|\iprod{U_{i-1}D_{i} \tilde{f}, \tilde{f}}| \leq \frac{i+1}{i} \gamma.$
\end{claim}

Assuming the above:
 \begin{multline*}
 \abs{\frac{\iprod{A_r f, f}}{\iprod{f,f}}} = |(i+1) \Dist{i-1}(r)\iprod{A_r f, f}| =
 (i+1)|\iprod{M^+_i \tilde{f}, \tilde{f}}| \leq \\
 (i+1)|\iprod{(M^+_i  - U_{i-1}D_{i})\tilde{f}, \tilde{f}}|  + (i+1)|\iprod{U_{i-1}D_{i} \tilde{f}, \tilde{f}}| \leq (i+1)\left(1 + \frac{i+1}{i}\right) \gamma \leq 3(i+1)\gamma,
 \end{multline*}
where the equalities in the first line use \eqref{eq:converse-1} and \eqref{eq:converse-2}, and the inequalities in the second line use \cref{f-down-operator-bounding-claim}, our assumption that $\|M^+_i  - U_{i-1}D_{i}\| \leq \gamma$, and the triangle inequality.
\end{proof}

We complete the proof of \cref{thm:equivalence} by proving \cref{f-down-operator-bounding-claim}:
\begin{proof}[Proof of \cref{f-down-operator-bounding-claim}]
Since $U_{i-1}D_i$ is PSD, we have $\iprod{U_{i-1}D_i \tilde{f}, \tilde{f}} \geq 0$, and so we may remove the absolute value and prove
\[ \iprod{U_{i-1}D_{i} \tilde{f}, \tilde{f}} \leq \frac{i+1}{i} \gamma.\]

Consider the inner product $\iprod{D_{i} \tilde{f}, D_{i} \tilde{f}}$ --- this is the expectation upon choosing $r' \sim \Dist{i-1}$, and then choosing two $i$-faces $s_1, s_2 \in X(i)$ containing it independently. Hence we decompose to the cases where $r' = r$ and $r' \ne r$:
\begin{gather}
\iprod{D_{i} \tilde{f}, D_{i} \tilde{f}} = \EE_{(r',s_1,s_2)}[f(s_1)f(s_2)] = \notag \\
\label{eq:down-op-1}
\Dist{i-1}(r)\EE_{(r',s_1,s_2)}[\tilde{f}(s_1)\tilde{f}(s_2) | r' = r] + (1-\Dist{i-1}(r))\EE_{(r',s_1,s_2)}[\tilde{f}(s_1)\tilde{f}(s_2) | r' \ne r].
\end{gather}
The first term is $0$, since from independence of $s_1, s_2$:
\[\EE_{(r',s_1,s_2)}[\tilde{f}(s_1)\tilde{f}(s_2) | r' = r] = \EE_{s_1}[f(s_1) | r \subset s_1]^2 = 0,\]
since by assumption $f$ is perpendicular to constant functions.

We saw above that for any $r' \ne r$, there is at most one $i$-face containing $r'$ (which is $s = r \cup r'$) such that $\tilde{f}(s) \ne 0$. For any $r' \ne r$, the value $\tilde{f}(s_1)\tilde{f}(s_2)$ is non-zero only when $s_1 = s_2 = r \cup r'$.
For every $s_1 \in X(i)$, we define the event $E_{s_1}$ to hold when $s_2 = s_1$. Namely, the event $E_{s_1}$ is the event where we choose an edge in the lower walk that is a triple $(s_1,r',s_1)$ where $s_1 \in X(i),r' \in X(i-1), r'\subseteq s_1$. \eqref{eq:down-op-1} is equal to
\[ \eqref{eq:down-op-1}=(1-\Dist{i-1}(r))\EE_{s_1}\bigl[\tilde{f}^2(s_1) \Pr_{r', s_2}[E_{s_1} | r' \ne r] \bigr].\]
Note that if $s_1$ doesn't contain $r$ then $\tilde{f}^2(s_1) = 0$, hence we continue taking expectation over all $s_1 \in X(i)$, even though some of them are unnecessary terms.

If we prove that for every $s \in X(i)$ we have $\Pr_{r', s_2}[E_{s_1} | r' \ne r] \leq \frac{i+1}{i}\gamma$ then
\[(1- \Dist{i-1}(r))\EE_{s_1}[\tilde{f}(s_1)^2 \Pr_{r', s_2}[E_{s_1}] ] \leq \frac{i+1}{i}\gamma \EE_{s_1}[\tilde{f}^2(s_1)] = \frac{i+1}{i}\gamma \iprod{\tilde{f},\tilde{f}} = \frac{i+1}{i}\gamma.\]
Thus we are left with proving the following statement: for all $s_1 \in X(i)$,
\[ \Pr_{r', s_2}[E_{s_1} | r' \ne r] \leq \frac{i+1}{i}\gamma.\]

We first bound the unconditioned probability $\Pr[E_{s_1}] = \Pr_{r',s_2 \in X(i)}[s_2 = s_1 | r' \subset s_1,s_2]$. Fix some $s_1 \in X(i)$, and let $\one_{s_1} \colon X(i) \to \RR$ be its indicator. Notice that $U_{i-1}D_{i}\one_{s_1}(s_1) = \Pr_{r', s_2}[s_2 = s_1]$, and so
\[ \iprod{U_{i-1}D_{i} \one_{s_1}, \one_{s_1}} = \Dist{i}(s_1) U_{i-1}D_{i}\one_{s_1}(s_1) = \Dist{i}(s_1) \Pr[E_{S_1}]. \]
We again use the non-laziness property of $M^+_i$ to assert that $\iprod{M^+_i \one_{s_1}, \one_{s_1}} = 0$.
Since $X$ is a $\gamma$-HDX,
\[ \iprod{U_{i-1}D_{i} \one_{s_1}, \one_{s_1}} = \iprod{(U_{i-1}D_{i} - M^+_i) \one_{s_1}, \one_{s_1}} \leq \norm{U_{i-1}D_{i} - M^+_i} \|\one_{s_1}\|^2 = \gamma \Dist{i}(s_1). \]
Hence $\Pr[E_{s_1}] \leq \gamma$.

Consider now any $s_1 \in X(i)$ containing $r$, and let $r'$ be a random $(i-1)$-face contained in $s_1$.
The probability that $r' \ne r$ is $\frac{i}{i+1}$, and so
\[ \Pr_{r', s_2}[E_{s_1} | r' \ne r]  \leq \frac{i+1}{i} \Pr_{r', s_2}[E_{s_1}] \leq \frac{i+1}{i}\gamma. \qedhere \]
\end{proof}

\subsection[]{Tightness of \cref{thm:equivalence}}
\label{sec:tightness-of-equivalnce-theorem}
In the remainder of the section we give two counterexamples that show that the dependence of \cref{thm:equivalence} on \(\lambda\) and \(d\) is tight up to constants.

To show the tightness of the first item in the theorem, namely, that there are simplicial complexes where \(\norm{M^+_i - UD} \approx \gamma\), it is enough to consider the complete \(d\)-dimensional simplicial complex on \(n\) vertices. On the one hand, the local links of the complete complex are complete graphs on \(n-d+1\) vertices (or more), thus it is a \(\frac{1}{n-d}\)-two sided link expander. By direct calculation,
\begin{align*}
 M^+_d - UD &= \frac{1}{(d+1)(n-d-1)} J(n,d+1) - \left(\frac{1}{n-d} I + \frac{1}{(d+1)(n-d)} J(n,d+1)\right) \\ &=
 \frac{1}{(d+1)(n-d-1)(n-d)} J(n,d+1) - \frac{1}{n-d} I,
\end{align*}
where $J(n,d+1)$ is the adjacency matrix of the Johnson graph. The eigenvalues of $J(n,d+1)$ are $(d+1-j)(n-d-1-j)-j$ for $j=0,\ldots,d+1$ (see \cite{BrouwerCN1989}), and so the eigenvalues of $M^+_d - UD$ are
\[
 \frac{(d+1-j)(n-d-1-j)-j}{(d+1)(n-d-1)(n-d)} - \frac{1}{n-d} =
 -\frac{j(n+1-j)}{(d+1)(n-d-1)(n-d)} \,.
\]
The spectral norm of $M^+_d - UD$ is attained at $j=d+1$, at which point it equals
\[
 \frac{(d+1)(n-d)}{(d+1)(n-d-1)(n-d)} = \frac{1}{n-d-1}.
\]
\medskip

As for the second item, we show a sequence of complexes \(X_n\) with a disconnected link (i.e.\ they are not \(\lambda\)-two sided link expanders for any \(\lambda < 1\)), where
\[\norm{M^+_d - UD} \leq \frac{1}{d} + o_n(1).\]

The complex \(X_n\) is obtained by removing a few faces from the \(d\)-dimensional complete complex, so that there is a cut in a single link. More formally, we define
\[X_n(0) = \set{1,2,\dots,n}\]
and
\[X_n(d) = \binom{[ n]}{d+1} \setminus \sett{\set{1,2,\dots,d-1} \cup \set{x,y}}{d\leq x<y \leq n, \; x+y\equiv1 \pmod{2}},\]
where \(\Dist{d}\) is uniform.

One can observe that the link of \(r_0 = \set{1,2,\dots,d-1}\) contains two connected components --- the even vertices and the odd vertices.

We comment that instead of removing the set of faces we can reduce their weight significantly so that the   link of $r_0$ will remain connected but not an expander. We next claim that the upper and lower walks are spectrally similar.
\begin{claim}
Let \(d > 1\), and let \(X_n\) be the sequence of simplicial complexes defined above. Then
  \[\norm{M^+_d - UD} \leq \frac{1}{d} + o_n(1).\]
\end{claim}

\begin{proof}
We fix \(n\) and denote \(X=X_n\), and for simplicity of computation we assume \(n-d\) is odd. We need to show that for any \(f\colon X(d-1)\to \mathbb{R}\) of norm \(1\),
\[ \abs{\iprod{(M^+_d -UD)f,f}} \leq \frac{1}{d}+o_n(1).\]

We decompose the inner product to local inner products on the link, as in \eqref{eq:localization-equation}:
\begin{equation} \label{eq:tightness-localization}
\abs{\iprod{(M^+_d-UD)f,f}} = \abs{\EE_{r \in X(d-1)}[\iprod{A_rf_r,f_r} - Df(r)^2]} \leq \EE_{r \in X(d-1)} [\lambda(X_r) \|f_r\|^2],
\end{equation}
where \(f_r\colon X_r(0)\to \RR\) is defined by \(f_r(x)=f(r \cup \set{x})\). For every \(r \ne r_0\) we have the following claim:
\begin{claim} \label{claim:other-links-expand}
  Let \(r \in X(d-2), \; r\ne r_0\), then \(\lambda(X_r)=o_n(1)\).
\end{claim}

Substituting this in~\eqref{eq:tightness-localization} and using $\EE_{r \in X(d-1)} \|f_r\|^2 = \|f\|^2 = 1$ and $\lambda(X_{r_0}) \leq 1$ gives
\[
 \abs{\iprod{(M^+_d-UD)f,f}} \leq o_n(1) + \Dist{}(r_0) \|f_{r_0}\|^2.
\]
The last term equals
\[
 \Dist{}(r_0) \|f_{r_0}\|^2 = \Dist{}(r_0) \sum_{x \in X_{r_0}(0)} w_{r_0}(x) f(r_0 \cup \{x\})^2 = \sum_{x \in X_{r_0}(0)} \frac{\Pi(r_0 \cup \{x\})}{d} f(r_0 \cup \{x\})^2 \leq \frac{1}{d}.
\]

We prove \cref{claim:other-links-expand} below.
\end{proof}

\begin{proof}[Proof of \cref{claim:other-links-expand}]
Denote \(m=n-d+1\), and without loss of generality, assume \(m\) is even.

Let \(r \ne r_0\). If \(|r \setminus r_0| \geq  3\), then \(X_r\) is the complete graph on $m$ vertices and is a \(\frac{1}{m-1}\)-two-sided spectral expander.

If \(|r \setminus r_0| = 2\) then \(X_r\) is one of the following:
\begin{enumerate}
    \item The complete graph --- when \(x,y \in r \setminus r_0\) have the same parity.
    \item A complete graph with only one edge missing --- when \(x,y \in r \setminus r_0\) have different parity.
\end{enumerate}
In both cases these are \(O \left (\frac{1}{m} \right )\) spectral expanders (skipping a short calculation in the second case).

If \(|r \setminus r_0| = 1\) then the graph \(X_r\) is obtained by taking the complete graph, and removing \(\frac{m}{2}\) of the edges that are adjacent to the vertex \(v_0 \in r_0 \setminus r\). We claim that this graph is still a \(o_n(1)\)-spectral expander.

The first vertex \(v_0\) is connected to \(\frac{m}{2}\) vertices, thus there are \(\frac{m}{2}-1\) vertices which are connected to all vertices besides themselves and \(v_0\), and \(\frac{m}{2}\) vertices that are connected to all vertices besides themselves.

The adjacency operator for this graph is a block matrix, formed by partitioning the rows and the columns into three parts, of sizes $1,\frac m2-1,\frac m2$, respectively:
\[
 A =
 \left(\begin{array}{c|c|c}
 0_{1 \times 1} & 0_{1 \times (\frac m2-1)} & \frac{2}{m} J_{1 \times \frac m2} \\\hline\Tstrut\Bstrut
 0_{(\frac m2-1) \times 1} & \frac{1}{m-2} K_{(\frac m2-1) \times (\frac m2-1)} & \frac{1}{m-2} J_{(\frac m2-1) \times \frac m2} \\\hline\Tstrut\Bstrut
 \frac{1}{m-1}J_{\frac m2 \times 1} & \frac{1}{m-1} J_{\frac m2 \times (\frac m2-1)} & \frac{1}{m-1} K_{\frac m2 \times \frac m2}
 \end{array}\right) \,,
\]
where $K$ is the adjacency matrix of a complete graph, and $J$ is the all-1 matrix.

The non-constant eigenvectors of $\frac{1}{m-2} K_{(\frac m2-1) \times (\frac m2-1)}$ (middle block) lift to an $(\frac m2-2)$-dimensional eigenspace of $-\frac{1}{m-2}$. Similarly, the non-constant eigenvectors of $\frac{1}{m-1} K_{\frac m2 \times \frac m2}$ (bottom-left block) lift to an $(\frac m2-1)$-dimensional eigenspace of $-\frac{1}{m-1}$. The remaining three eigenvalues correspond to eigenvectors which are constant on blocks. A straightforward calculation shows that these eigenvalues are roots of the cubic polynomial
\[
 (4m^2-12m+8) \lambda^3 - (4m^2 - 18m + 16)\lambda^2 - (8m - 16) \lambda + (2m - 8) = 0.
\]
The trivial eigenvalue $\lambda = 1$ corresponds to the constant eigenvector. The other two are roots of the quadratic
\[
 (4m^2-12m+8) \lambda^2 + (6m-8) \lambda - (2m - 8) = 0.
\]
The roots of this quadratic are
\[
 \lambda = \frac{-6m+8 \pm \sqrt{32m^3 - 188m^2 + 352m - 192}}{8m^2-24m+16} = \pm\Theta\left(\frac{1}{\sqrt{m}}\right). \qedhere
\]
\end{proof}

\section{Expanding posets (eposets)} \label{sec:eposet}

In this section, we describe a setting generalizing simplicial complexes, namely \emph{measured posets}. These are partially ordered sets (a set $X$ with a partial order $\leq$ on it) whose elements are partitioned into levels $X(j)$, and that have some additional properties stated below. As in simplicial complexes, we can define $C^j$ as the space of real-valued functions on $X(j)$, and averaging operators $U_j\colon C^{j} \to C^{j+1}$ and  $D_{j+1}\colon C^{j+1} \to C^{j}$.

We generalize the notion of a $\gamma$-HDX to a $\gamma$-\emph{expanding poset} (eposet) --- a measured poset with operators $D_j,U_j$ such that
\[ \norm{D_{j+1}U_j - r_jI - \delta_j U_{j-1}D_j } \leq \gamma, \]
for $\gamma <1$, all non-extreme levels $j$ of the poset, and some constants $r_j, \delta_j$.

We begin the section by discussing the formal notion of an eposet. We then generalize \cref{thm:approximate-orthogonality-hdx} to all eposets, and prove it in the general setting. Finally, we show that if our measured poset is a simplicial complex, then $r_j \approx \frac{1}{j+2}, \delta_j \approx 1-\frac{1}{j+2}$, under the assumption that the laziness of the lower walk is small.

\subsection{Measured posets} \label{sec:measured}

A \emph{graded} (or \emph{ranked}) \emph{poset} is a partially ordered set (poset) $(X,\leq)$ equipped with a \emph{rank} function $\rho\colon X \to \mathbb{N} \cup \{0,-1\}$ such that:
\begin{enumerate}
\item For all $x,y \in X$, if $x \leq y$ then $\rho(x) \leq \rho(y)$.
\item For every $x,y \in X$, if  $y$ is minimal with respect to elements greater than $x$ (i.e.\ $x\leq y$), then $\rho(y) = \rho(x) + 1$.
\end{enumerate}
We denote the set of elements of rank $j$ by $X(j)$. We assume that there is a unique element of minimal rank which we denote by $\emptyset$, and so $X(-1) = \{\emptyset\}$.

We say that a graded poset is \emph{$d$-dimensional} if the maximal rank of any element in $X$ is $d$. We say that a $d$-dimensional graded poset is \emph{pure} if all maximal elements are of rank $d$, that is, for every $t \in X$ there exists $s \in X(d)$ such that $t \leq s$.

For example, any simplicial complex is a graded poset, if we take $\leq$ to be the containment relation and $\rho$ to be the cardinality of a face, minus one.
Another useful example to keep in mind is the \emph{Grassmann poset} $\Gr_q(n,d)$, whose elements are subspaces of dimension at most $d+1$ of $\mathbb{F}_q^n$, and the order is by containment. The rank function for the Grassmann poset is $\rho(U) = \dim(U) - 1$, and so $X(j) = \{U \subseteq \mathbb{F}_q^n : \dim(U) = j+1 \}$.

\begin{definition}[Measured poset]
Let $X$ be a finite graded pure $d$-dimensional poset, with a unique minimum element $\emptyset$ of rank $-1$. We say that $X$ is \emph{measured} by a (joint) distribution $\Distv = (\Dist{d},\Dist{d-1},\dots,\Dist{-1})$ if it satisfies the following properties\footnote{As is common, we abuse notation $\Pi_i$ to refer to both the distribution as well as the random variable sampled according to the distribution.}:
\begin{enumerate}
\item $\Dist{i} \in X(i)$ for all $i$.
\item $\Dist{i-1} \subset \Dist{i}$ for all $i > -1$.
\item The sequence $\Dist{d},\dots,\Dist{-1}$ has the Markov property: $\Dist{i-1}$ depends only on $\Dist{i}$ for all $i > -1$.
\end{enumerate}
We denote the real-valued function spaces on $X(j)$ by $C^j$. We denote the averaging operators of the steps in the Markov process by $U_j\colon C^j \to C^{j+1}, D_{j+1}\colon C^{j+1} \to C^j$.
\end{definition}
The operators $U_j$ and $D_{j+1}$ are adjoint with respect to the inner product given by \[
\iprod{f,g} = \E_{x\sim \Pi_j}f(x)g(x) \]
and thus both $U_jD_{j+1}$ and $D_{j+1}U_{j}$ are positive semi-definite since, for example, for all $f$ \[\iprod{U_jD_{j+1}f,f} = \iprod{D_{j+1}f,D_{j+1}f} \ge 0.\]
The process we defined for the distribution $\Distv$ in a simplicial complex is an example of a measured poset. For the Grassmann poset mentioned above, we also have a similar probabilistic experiment:
\begin{enumerate}
	\item Choose a subspace of dimension $d+1$, $s_d \in X(d)$, uniformly at random.
	\item Given a subspace $s_i$ of dimension $i+1$, choose $s_{i-1} \in X(i-1)$ to be a uniformly random codimension~1 subspace of $s_i$.
\end{enumerate}

\medskip

An analog for \cref{thm:decomposition} holds for any measured poset. We say that a $k$-dimensional measured poset $X$ is \emph{proper} if for all $j \leq k-1$, $\ker U_j = \{0\}$. Also, as before we denote
\[H^{-1} = C^{-1}, \quad H^i = \ker D_i, \quad V^i = U^{k-i}H^i. \]

\begin{theorem} \label{thm:decomposition-asd}
If $X$ is a proper $k$-dimensional measured poset then we have the following decomposition of $C^k$:
\[ C^k = V^k \oplus V^{k-1} \oplus \cdots \oplus V^{-1}\;. \] In other words, for every function $f\in C^k$ there is a unique choice of $h_i \in H^i$ such that the functions $f_i = U^{k-i}h_i$ satisfy $f = f_{-1} + f_{0} + \ldots + f_k$.
\end{theorem}
\begin{proof}
We first prove by induction on $\ell$ that every function $f \in C^\ell$ has a representation $f = \sum_{i=-1}^\ell U^{\ell-i} h_i$, where $h_i \in H^i$. This trivially holds when $\ell = -1$. Suppose now that the claim holds for some $\ell<k$, and let $f \in C^{\ell+1}$. Since $D^{\ell+1}\colon C^{\ell+1} \to C^\ell$ is a linear operator, we can decompose $C^{\ell+1}$ to $\ker D_{\ell+1} \oplus (\ker D_{\ell+1})^\bot$. It is well known that $(\ker D_{\ell+1})^\bot = \im D_{\ell+1}^\ast$ so we have $C^{\ell+1} = \ker D_{\ell+1} \oplus \im D_{\ell+1}^\ast = \ker D_{\ell+1} \oplus \im U_\ell$, and therefore we can write $f = h_{\ell+1} + Ug$, where $h_{\ell+1} \in H^{\ell+1}$ and $g \in C^\ell$. Applying induction, we get that $g = \sum_{i=-1}^\ell U^{\ell-i} h_i$, where $h_i \in H^i$. Substituting this in $f = h_{\ell+1} + Ug$ completes the proof.

It remains to show that the representation is unique. Since $\ker U_{i-1} = \ker D_i^*$ is trivial, $\dim H^i = \dim C^i - \dim C^{i-1}$ for $i \geq 0$. This shows that $\sum_{i=-1}^k \dim H^i = \dim C^k$. Therefore the operator $\varphi\colon H^{-1} \times \cdots \times H^k \to C^k$ given by $\varphi(h_{-1},\ldots,h_k) = \sum_{i=-1}^k U^{k-i} h_i$ is not only surjective but also injective. In other words, the representation of $f$ is unique.
\end{proof}

\subsection{Sequentially differential posets}
\emph{Sequentially differential posets} were first defined and studied (in a slightly different form) by Stanley~\cite{Stanley1988,Stanley1990}.
\begin{definition}[Sequentially differential posets] \label{def:SD}
\emph{Sequentially differential posets} are measured posets whose averaging operators $U$, $D$ satisfy an equation
\begin{equation}\label{eq:sd} D_{j+1}U_j - \delta_j U_{j-1}D_j - r_j I = 0, \end{equation}
for some $r_j, \delta_j \in \RR_{\geq 0}$ and all $0 \leq j \leq k-1$.
\end{definition}

For example, the complete complex satisfies this definition with parameters
\[ \delta_i = \left (1-\frac{1}{i+2}\right )\left (1-\frac{1}{n-i} \right)^{-1} \text{ and } r_i = 1- \delta_i. \]
In other words,
\[ DU - \left (1-\frac{1}{i+2}\right)\left (1-\frac{1}{n-i}\right)^{-1} UD - \left (1 - \left (1-\frac{1}{i+2}\right )\left (1-\frac{1}{n-i}\right )^{-1} \right) I = 0. \]

The Grassmann poset $\Gr_q(n,d)$ is also a sequentially differential poset with
\begin{equation}\label{eq:Grassmann-SD}
  \delta_i = 1 - \left (1-\frac{q-1}{q^{i+2}-1}\right) \left (1 - \frac{q-1}{q^{n-i}-1}\right)^{-1} \text{ and } r_i = 1-\delta_i.
\end{equation}

The above following from the claim below, that the reader can verify by direct calculation:
\begin{claim} \label{claim:SD-criterion}
Let $X$ be a measured poset, and suppose we can decompose:
\[D_{i+1}U_i = \alpha_iI + (1-\alpha_i)M_i,  \]
\[U_{i-1}D_i = \beta_iI + (1-\beta_i)M_i,  \]
where $0 \leq \alpha_i, \beta_i \leq 1$ are constants and $M_i$ is some operator.
Then
\[ D_{i+1}U_i - r_i I - \delta_i U_{i-1}D_i = 0, \]
where
\[ \delta_i = (1-\alpha_i)(1-\beta_i)^{-1} \text{ and } r_i = 1- \delta_i. \]
\end{claim}

In both the complete complex and the Grassmann poset $\Gr_q(n,d)$, the non-lazy upper walk and the non-lazy lower walk are \emph{the same} --- given $t_1 \in X(i)$, our choice for $t_2 \in X(i)$ is a set (or subspace in the Grassmann case) that shares an intersection of size (resp.\ dimension) $i$ with $t_1$ (with uniform probability). The only difference between $DU$ and $UD$ is the probability to stay in place.
Thus we can decompose:
\[D_{i+1}U_i = \alpha_iI + (1-\alpha_i)M_i,  \]
\[U_{i-1}D_i = \beta_iI + (1-\beta_i)M_i,  \]
where $M_i$ is the non-lazy upper (or lower) random walk. In the simplicial complex case
\[\alpha_i = \frac{1}{i+2}, \; \beta_i = \frac{1}{n-i},\]
and in the Grassmann case
\[\alpha_i = \frac{q-1}{q^{i+2}-1}, \; \beta_i = \frac{q-1}{q^{n-i}-1}.\]

\medskip

We relax \cref{def:SD} to an {\em almost} sequentially differential poset --- a measured poset that approximately satisfies such an identity:
\begin{definition}[Expanding Poset]\label{def:eposet}
Let $\vec{r}, \vec{\delta} \in \RR_{\geq 0}^{k}$, and let $\gamma < 1$. We say that $X$ is an \emph{$(\vec{r},\vec{\delta},\gamma)$-expanding poset} (or \emph{$(\vec{r},\vec{\delta},\gamma)$-eposet}) if for all $j \leq k-1$:
\begin{equation}\label{eq:asd}
\norm{D_{j+1} U_j - r_j I - \delta_j U_{j-1} D_j} \leq \gamma.
\end{equation}
\end{definition}
A sequentially differential poset is an eposet with $\gamma = 0$. As we saw in \eqref{eq:SD-def-HDX}, a $\gamma$-HDX is an $(\vec{r}, \vec{\delta}, \gamma)$-eposet, where $r_j = \frac{1}{j+2}$ and $\delta_j = 1 - \frac{1}{j+2}$.

We can use \eqref{eq:Grassmann-SD} to assert that $\Gr_q(n,d)$ is an $(\vec{r},\vec{\delta},\gamma)$-eposet for $r_i = \frac{q-1}{q^{i+2}-1}, \delta_i = 1 - r_i$, and $\gamma = O(1/q^{n-d})$. While this only shows that the $\Gr_q(n,d)$ is an \emph{eposet} (even though it is truly sequentially differential), the parameters are much simpler, thus calculations regarding the random walks are easier (see for instance the calculations in \cref{sec:grassmann}).

\subsection{Almost orthogonality of decomposition} \label{sec:orthogonality}

In this section we show that in an eposet, the spaces $V_i$ are almost orthogonal to one another. Moreover, we show that these spaces are ``almost eigenspaces'' of the operator $DU$.

\begin{theorem} \label{thm:approximate-orthogonality-eposets}

Let $X$ be a $k$-dimensional $(\vec{r},\vec{\delta},\gamma)$-eposet.
For every function $f$ on $C^\ell$ for $\ell \leq k$, the decomposition $f = f_{-1} + \cdots + f_\ell$ of \cref{thm:decomposition-asd} satisfies the following properties, when $\gamma$ is small enough (as a function of $k$ and the eposet parameters):
\begin{itemize}
\item For $i \neq j$, $|\langle f_i, f_j \rangle| = O(\gamma) \|f_i\| \|f_j\|$.
\item $\|f\|^2 = (1 \pm O(\gamma))(\|f_{-1}\|^2 + \cdots + \|f_\ell\|^2)$, and for all $i$, $\|f\|^2 = (1 \pm O(\gamma))(\|f_{\leq i}\|^2 + \|f_{>i}\|^2)$.
\item If $\ell < k$, the $f_i$ (for $i \geq 0$) are approximate eigenvectors of $DU$: $\|DUf_i - r^\ell_{\ell-i+1} f_i\| = O(\gamma) \|f_i\|$, where

\begin{equation} \label{eq:eigenvalues-of-DU}
r^\ell_i = r_\ell + \sum_{j=\ell-i}^{\ell-1} \left ( \prod_{t=j+1}^\ell \delta_t \right) r_j .
\end{equation}
(Note $DUf_{-1} = r^\ell_{\ell+2} f_{-1}$, where $r^\ell_{\ell+2} = 1$.)

\item If $\ell < k$ then $\langle DUf,f \rangle = \sum_{i=-1}^\ell r^\ell_{\ell-i+1} \|f_i\|^2 \pm O(\gamma)\|f\|^2$.
\end{itemize}
The hidden constant in the $O$~notations depends only on $k$ and the eposet parameters $(\vec{r},\vec{\delta},\gamma)$ and not on the the eposet size $|X|$. In particular, the last item implies that if $\vec{r} > 0$ then for a small enough $\gamma$, the poset is proper.
\end{theorem}
In a measured poset, the decomposition of \cref{thm:decomposition-asd} is not necessarily orthogonal. However, this theorem shows that for an eposet, the decomposition is \emph{almost} orthogonal.

\begin{remark}
In the special case of a \emph{sequentially differential poset}, i.e.\ $\gamma = 0$, we \emph{do} get that the decomposition in \cref{thm:decomposition-asd} is orthogonal, and that the decomposition $C^\ell = V_{-1} \oplus \cdots \oplus V_\ell$ is a decomposition to eigenspaces of $DU$: for all $f_i \in V_i$,
\[DU f_i = r^\ell_i f_i, \]
for the $r^\ell_i$ given in \eqref{eq:eigenvalues-of-DU}.
\end{remark}

\medskip

Recall our convention that for $f \in C^{\ell-j}$,
\[ U^jf = U_{\ell-1}\cdots U_{\ell-j+1}U_{\ell-j}f  \in C^\ell .\]

We first show how the third item in \cref{thm:approximate-orthogonality-eposets} implies the rest. The following proposition says that the decomposition in \cref{thm:decomposition-asd} is a decomposition of ``approximate eigenspaces" of $UD$. We postpone its proof to the end of this section, and use it first to obtain the full statement of \cref{thm:decomposition-asd}. Recall that $H^{\ell-j} = \ker D_{\ell-j} \subseteq C^{\ell-j}$.
\begin{proposition} \label{prop:approx-eigenvalues}
    Let $X$ be an $(\vec{r}, \vec{\delta}, \gamma)$-eposet, and let $h \in H^{\ell-j}$. Then $U^j h \in V^{\ell-j}$ is an approximate eigenvector of $D_{\ell+1} U_\ell$ with eigenvalue $r^\ell_{j+1}$:
    \[ \|DU (U^j h) - r^\ell_{j+1} (U^j h)\| = O( \gamma) \norm{h} . \]
\end{proposition}

We proceed by showing that these approximate eigenspaces $V^j$ are approximately orthogonal.

\begin{lemma} \label[lemma]{lem:approximate-orthogonality}
Suppose that $X$ is a $k$-dimensional $(\vec{r}, \vec{\delta}, \gamma)$-eposet, let $\ell < k$, let $i \ne j$, and let $f_i = U^{\ell-i} h_i, f_j = U^{\ell-j} h_j$ for $h_i \in H^i, h_j \in H^j$, as in \cref{thm:decomposition}. Then
\[ \iprod{f_i, f_j } = O(\gamma) \norm{h_i} \norm{h_j}, \]
where the hidden constant depends on $k,\vec{\delta},\vec{r}$ only.
\end{lemma}

\begin{proof}
  Assume without loss of generality that $i>j$. To prove the statement we use \cref{prop:approx-eigenvalues} and induction on $m=\ell-i$. The base case where $m=\ell-i=0$ (or $\ell=i$) follows from the fact that $h_\ell$ is orthogonal to $f_{\ell-j}$, for any $j \geq 1$ since it is in $H^\ell = \ker D_\ell$. Indeed, $\iprod{f_\ell,f_{\ell-j}} = \iprod{h_\ell,U^j h_{\ell-j}} = \iprod{D^j h_\ell,h_{\ell-j}}=0$.

  Assuming the statement holds for $m$, we show it for $m+1$. Let $f_i,f_j$ be as above (where $\ell-i=m+1$). Then
  \[ \iprod{f_i, f_j} = \iprod{U^{\ell-i} h_i, U^{\ell-j} h_j} = \iprod{DU^{\ell-i} h_i, U^{(\ell-1)-j} h_j} .\]
  As $\norm{DU^{\ell-i} h_i - r^\ell_{\ell-i+1} U^{(\ell-1)-i}h_i} = O(\gamma) \norm{h_i}$ then by Cauchy-Schwarz this upper bounded by
  \[O(\gamma) \norm{h_i} \norm{U^{(\ell-1)-j} h_j} + r^\ell_{(\ell-1)-i} \iprod{U^{(\ell-1)-i}h_i, U^{(\ell-1)-j} h_j}.\]
  By the induction hypothesis on $m=\ell-1-i \geq 0$ this less or equal to
  \[O(\gamma)\norm{h_i}\norm{U^{(\ell-1)-j} h_j} + r^\ell_{(\ell-1)+i} O(\gamma) \norm{h_i} \norm{h_j}.\]
  The Up operator is an averaging operator so $\norm{U^{(\ell-1)-j} h_j} \leq \norm{h_j}$ and the lemma follows.
\end{proof}

The preceding lemma gives an error estimate in terms of the norms $\norm{h_i}$. The following lemma enables us to express the error in terms of the norms $\norm{f_i}$.

\begin{lemma} \label[lemma]{lem:approximate-U}
For any $k$-dimensional $(\vec{r}, \vec{\delta}, \gamma)$-eposet, let $\ell < k$ and let $f_i = U^{\ell-i} h_i$ for $h_i \in H^i$, as in \cref{thm:decomposition}. Then
\[ \norm{f_i} = (1 \pm O(\gamma)) \rho^\ell_{\ell-i} \norm{h_i}, \]
where $\rho^\ell_j = \prod_{t=0}^{j} r^{\ell-t}_{j-t}$, and the hidden constant depends only on $k, \vec{r},\vec{\delta}$.
\end{lemma}

\begin{proof}
By direct calculation with \cref{prop:approx-eigenvalues} we obtain that for any $h \in \ker D$:
\[D^jU^j h =r^\ell_j D^{j-1}U^{j-1} h + \Gamma_1 = \dots =  \rho^\ell_j h + \sum_{t=1}^j \Gamma_t, \]
where $\Gamma_t$ is the remainder, and $\norm{\Gamma_t} = O(\gamma) \norm{h}$ for all $t$. Thus
\[\norm{D^jU^j h - \rho^\ell_j h} = O(\gamma)\norm{h}.\]
Hence using Cauchy--Schwarz,
\[\|f_i\|^2 = \iprod{U^{\ell-i} h_i, U^{\ell-i} h_i} = \iprod{D^{\ell-i} U^{\ell-i} h_i, h_i} = \rho^\ell_{\ell-i} \|h_i\|^2 \pm O(\gamma) \|h_i\|^2. \qedhere \]
\end{proof}

Combining \cref{lem:approximate-orthogonality} and \cref{lem:approximate-U}, we obtain the following corollary, which proves the first item of \cref{thm:approximate-orthogonality-eposets}.

\begin{corollary} \label[corollary]{cor:approximate-orthogonality}
Suppose that $X$ be a $k$-dimensional $(\vec{r}, \vec{\delta}, \gamma)$-eposet, let $\ell < k$, and let $f \in C^\ell$ have the decomposition $f = f_{-1} + \cdots + f_\ell$, as in \cref{thm:decomposition}. Then for $i \neq j$ and small enough $\gamma$,
\[ \langle f_i, f_j \rangle = O(\gamma) \|f_i\| \|f_j\|, \]
where the hidden constant depends only on $k, \vec{r}, \vec{\delta}$.
\end{corollary}

As a consequence, we obtain an \emph{approximate $L_2$ mass formula}, constituting the second item of \cref{thm:approximate-orthogonality-eposets}:

\begin{corollary} \label[corollary]{cor:approximate-l2-mass}
Under the conditions of \cref{cor:approximate-orthogonality}, for every $i \leq j$ we have
\[
 \|f_i + \cdots + f_j\|^2 = (1 \pm O(\gamma)) (\|f_i\|^2 + \cdots + \|f_j\|^2),
\]
where the hidden constant depends only on $k, \vec{r}, \vec{\delta}$.

In particular,
\[
 \|f\|^2 = (1 \pm O(\gamma)) (wt_{\leq i}(f) + wt_{>i}(f)) = (1 \pm O(\gamma))(\|f_{\leq i}\|^2 + \|f_{>i}\|^2).
\]
\end{corollary}
\begin{proof}
Expanding $\|f_i + \cdots + f_j\|^2$, we obtain
\begin{multline*}
\bigl|\|f_i + \cdots + f_j\|^2 - \|f_i\|^2 - \cdots - \|f_j\|^2\bigr| \leq 2\sum_{i \leq a < b \leq j} |\langle f_a,f_b \rangle| = O(\gamma) \sum_{i \leq a < b \leq j} \|f_a\| \|f_b\| \leq \\
O(\gamma) \left(\|f_i\| + \cdots + \|f_j\|\right)^2 \leq
O(\gamma) (\|f_i\|^2 + \cdots + \|f_j\|^2),
\end{multline*}
swallowing a factor of $j-i+1$ in the last inequality.
\end{proof}

The fourth item of \cref{thm:approximate-orthogonality-eposets} follows from the preceding ones:

\begin{corollary} \label[corollary]{cor:approximate-laplacian}
Under the conditions of \cref{cor:approximate-orthogonality},
\[ \langle DUf,f \rangle = (1 \pm O(\gamma)) \sum_{i=-1}^\ell r^\ell_{\ell-i+1} \|f_i\|^2. \]	
\end{corollary}
\begin{proof}
 Let $DUf_i = r^\ell_{\ell-i+1} f_i + g_i$, where $\|g_i\| = O(\gamma) \|f_i\|$ according to the third item. Then
\begin{equation} \label{eq:approximate-laplacian}
 \iprod{DUf,f} = \sum_{i=-1}^\ell r^\ell_{\ell-i+1} \iprod{f_i,f} + \sum_{i=-1}^\ell \iprod{g_i,f}.
\end{equation}

 We can bound the magnitude of the second term using Cauchy--Schwarz:
\[
 \sum_{i=-1}^\ell |\iprod{g_i,f}| \leq \sum_{i=-1}^\ell \|g_i\| \|f\| = O(\gamma) \sum_{i=1}^\ell \|f_i\| \|f\| = O(\gamma) \|f\|^2,
\]
 using the second item.

 For every $i$, we can bound $\iprod{f_i,f}$ by
\[
 \iprod{f_i,f} = \|f_i\|^2 + \sum_{j \neq i} \iprod{f_i,f_j} = \|f_i\|^2 \pm O(\gamma) \|f_i\| \|f\|,
\]
 using the first two items.

 Substituting both bounds in~\eqref{eq:approximate-laplacian} and using the second item again, we get
\[
 \iprod{DUf,f} = \sum_{i=-1}^\ell r^\ell_{\ell-i+1} \|f_i\|^2 \pm O(\gamma) \|f\|^2. \qedhere
\]
\end{proof}

We turn to proving \cref{prop:approx-eigenvalues}. It follows directly from a technical claim that generalizes the approximate relation between $D$ and $U$, namely
\[\|DU - rI - \delta UD\| = O(\gamma),\]
 to an approximate relation between $D$ and $U^j$:
\[\|DU^j - r U^{j-1} - \delta U^jD\| = O(\gamma),\]
for appropriate constants $r,\delta \in \mathbb{R}$.

\begin{claim} \label[claim]{clm:generalized-D-U-relation}
Let $X$ be a $k$-dimensional $(\vec{r}, \vec{\delta}, \gamma)$-eposet, $1 \leq j \leq \ell+1 \leq k$, and $DU^j\colon X(\ell-(j-1)) \to X(\ell)$. There exist constants $r^\ell_j, \delta^\ell_j$ (as given below) such that

\begin{equation}\label{eq:generalized-D-U-relation}
\norm{DU^j - r^\ell_{j} U^{j-1} - \delta^\ell_{j} U^j D} = O(\gamma),
\end{equation}
where the hidden constant in the $O(\cdot)$ notation depends only on $k,\vec{\delta},\vec{r}$.

The constants $\delta^\ell_j$ and $r^\ell_j$ are given by the following formulas: $\delta_0^\ell = 1$ and
\begin{align*}
 \delta^\ell_j &= \prod_{t = \ell-(j-1)}^\ell \delta_t, &
 r^\ell_j &= \sum_{t = 0}^{j-1} r_{\ell-t} \delta^\ell_t.
\end{align*}
\end{claim}

Regarding the constants $r, \delta$, notice the following:
\begin{enumerate}
	\item $r^\ell_1 = r_\ell$ and $\delta^\ell_1 = \delta_\ell$.
	\item If for all $0 \leq j \leq \ell$, $r_j + \delta_j = 1$, then for all $0 \leq j \leq \ell$, $r^\ell_j + \delta^\ell_j = 1$.  In this case, we have a better formula for $r^\ell_j$:
	\[ r^\ell_j = 1 - \prod_{t = \ell-(j-1)}^\ell \delta_t. \]
	\item In a $\gamma$-HDX, we get $r^\ell_j = \frac{j}{\ell+2}$ and $\delta^\ell_j = 1-\frac{j}{\ell+2}$.
\end{enumerate}

\begin{proof}[Proof of \cref{prop:approx-eigenvalues}]
By \cref{clm:generalized-D-U-relation}
\[\norm{(DU^j - r^\ell_{j} U^{j-1} - \delta^\ell_{j} U^j D)h} = O(\gamma)\norm{h}.\]
$Dh = 0$ so the right-hand side is
$\norm{(DU^j - r^\ell_{j} U^{j-1} )h}$ and the proposition follows.
\end{proof}

Finally, we prove \cref{clm:generalized-D-U-relation}. While the statement of this claim seems technical, its proof consists of simply inductively substituting $DU$ with $rI + UD$ in the terms, until the formula is obtained.
\begin{proof}[ Proof of \cref{clm:generalized-D-U-relation}]
We prove the claim by induction on $j$. The base case $j=1$ follows by the definition of an eposet: $\delta^\ell_1 = \delta_\ell$, $r^\ell_1 = r_\ell$, and
\[ \norm{DU - r^\ell_1 I - \delta^\ell_1 UD} \leq \gamma .\]

For the induction step on $j+1$, note that $DU^{j+1} = D U^j U$. We add and subtract:
\begin{equation} \label{eq:generalized-relation-induction-step}
 DU^jU = \left[ DU^jU - (r^\ell_j U^{j-1} U + \delta^\ell_j U^j D U ) \right]  + (r^\ell_j U^j + \delta^\ell_j U^j D U ).
\end{equation}
The term inside the square brackets has spectral norm  at most $O(\gamma)\|U\|$ due to the induction hypothesis. Since $\|U\| \leq 1$,
\[ \|DU^jU - (r^\ell_j U^{j-1} U + \delta^\ell_j U^j D U)\| = O(\gamma).\]

We consider next the term $\delta^\ell_j U^j D U$, and substitute the $D U$ in it with
\[(r_{\ell-j}I + \delta_{\ell-j} UD) + \Gamma, \]
where $\Gamma = DU - (r_{\ell-j}I + \delta_{\ell-j} UD)$ has norm at most $\gamma$ (recall that the assumption that the poset is an \((\vec{r},\vec{\delta},\gamma)\)-eposet explicitly bounds the operator norm of \(\Gamma\) by \(\gamma\)). We get that
\[ \norm{ \delta^\ell_j U^j D U - \delta^\ell_j U^j  (r_{\ell-j}I + \delta_{\ell-j} UD) } = O(\gamma). \]
We rearrange the left-hand side of the equation to get
\[ \delta^\ell_j U^j D U - \delta^\ell_j U^j  (r_{\ell-j}I + \delta_{\ell-j} UD) = \delta^\ell_j U^j D U - r_{\ell-j} \delta^\ell_j U^j  - \delta^\ell_{j+1} U^{j+1}D. \]
Plugging this term back in \eqref{eq:generalized-relation-induction-step}, we get
\[ \|DU^{j+1} - r^\ell_{j+1} U^j - \delta^\ell_{j+1} U^{j+1} D\| = O(\gamma). \qedhere \]
\end{proof}

\subsection{Equivalence between link expansion and random-walk expansion for decomposable posets}\label{sec:poset-equiv}
In this subsection we extend the equivalence theorem of \cref{sec:HD-expanders} to a class of $\gamma$-eposets that share some key properties with simplicial complexes. We call these posets \emph{decomposable posets}.

We begin with the definition of a link in a general measured poset.

\begin{definition}
\label{def:link-of-poset}
Let \(X\) be a \(d\)-dimensional measured poset. Let \(s \in X(i)\). The link \(X_s\) is a \((d-i-1)\)-graded poset consisting of all $t \in X$ such that \(t \geq s\), with rank function \(\rho_s(y)=\rho(y)-\rho(x)-1\).

The induced (joint) distribution on the link \(\Distv_{X_s} = (\Dist{X_s,d-i-1},\dots,\Dist{X_s,-1})\) is defined as follows:
\[\Pr[\Distv_{X_s}= (t_{d-i},t_{d-i-1},\dots,t_1,t_0)]=\Pr[\Dist{d}=t_{d-i}, \Dist{d-1}=t_{d-i-1},\dots, \Dist{i+2}=t_1, \Dist{i+1} = t_0  \mid \Dist{i} = s].\]
\end{definition}
Namely, the probability of sampling \(t\) in \(X_s\) is the probability of sampling it given that \(s\) was sampled from the \(i\)-th level in \(X\).

We denote by \(U^s_j, D^s_j\) the upper and lower walks on \(X_s\) starting from $X_s(j)$. We further denote by \(M^{+,s}_j\) the non-lazy upper walk on \(X_s\) starting from $X_s(j)$.
When \(s = \emptyset\), that is, \(X_s = X\), we simply write \(M^+_j\).

We define a two-sided link expander poset analogously to the definition for simplicial complexes:
\begin{definition}
Let \(X\) be a measured poset. We say that \(X\) is a \(\gamma\)-two-sided-link expander if for every \(i \leq d-2\) and every \(s \in X(i)\), it holds that
\[ \lambda(M^{+,s}_0) \leq \gamma,\] where
$\lambda(M^{+,s}_0)$ is the second largest eigenvalue of $M_0^{+,s}$ in absolute value, which is also equal to $\norm{M^{+,s}_{0} - U^s_{-1} D^s_0} $.
\end{definition}

Our main theorem is that for a special class of measurable posets called \emph{decomposable} posets, the above definition is an equivalent characterization of an eposet. To that end, we first show that \cref{def:eposet} has an alternate characterization (see \cref{def:alteposet}) if the laziness is small.

\subsubsection{Laziness and an alternate characterization of eposets}

In this section, we show that if the \emph{laziness} of upper and lower walks of the eposet is small, then there is an alternate more convenient characterization of eposets in terms of $\norm{U_{i-1}D_i - M^+_i}$. First for some definitions.

\begin{definition}[laziness of eposet]
  Let $M$ be a random walk on the set $V$. We say that $M$ is $\alpha$-lazy for some $\alpha \in (0,1)$ if for every $t \in V$ we have $M(t,t) \leq \alpha$. If furthermore, the operator $M$ can be decomposed as
  \[ M = \alpha I + (1-\alpha) M^+,\]
  then we say that $M$ is $\alpha$-uniformly lazy.  In other words, the walk $M$ is an $(\alpha, 1-\alpha)$ convex combination of the lazy component $I$ and non-lazy component $M^+$.

  Let $X$ be a measured eposet. We say that the upper walk $DU$ is $\vec{\alpha}$-uniformly lazy for some  vector $\vec{\alpha} = (\alpha_0,\alpha_1,\dots, \alpha_{d-1})$ if each of the upper walks $D_{i+1}U_i$ are $\alpha_i$-uniformly lazy.

    If $\alpha_i \leq \alpha$ for all $i \geq 0$ for some $\alpha \in (0,1)$, we then say that the upper walk of $X$ is $\alpha$-uniformly lazy.
    \end{definition}
For example, that a $d$-dimensional simplicial complex is \(\vec{\alpha}\)-uniformly lazy where $\vec{\alpha}=(\frac{1}{2},\frac{1}{3},...\frac{1}{d+1})$.

\begin{lemma}\label{lem:eposets-defs}
 Let $\gamma \in (0,\nicefrac1{2})$.
  Let \(X\) be a \(d\)-dimensional measured poset whose lower walk $UD$ is $\gamma$-lazy and whose upper walk is $\nicefrac12$-uniformly lazy.
  Then
  \begin{enumerate}
  \item \label{partone} If $\norm{U_{i-1}D_i - M^+_i} \leq \gamma$ then $X$ is an $(\vec{r},\vec{\delta},\gamma)$-eposet for some $\vec{r}, \vec{\delta}$.
  \item \label{parttwo} If $X$ is an $(\vec{r},\vec{\delta},\gamma)$-eposet for some $\vec{r}, \vec{\delta}$, then $\norm{U_{i-1}D_i - M^+_i} = O(\gamma)$ for all $i \geq 0$.
    \end{enumerate}
\end{lemma}
\begin{proof}
Since the upper walk of $X$ is $\nicefrac12$-uniformly lazy, there exists $\vec{\alpha}=  (\alpha_0,\alpha_1,\dots, \alpha_{d-1})$ such that for all $i \in \set{-1, 0,\dots,d-1}$, we have
\[D_{i+1}U_i = \alpha_i I + (1-\alpha_i) M_i^+.\]
\paragraph{(Proof of Part~\ref{partone})} Suppose $ \norm{U_{i-1}D_i - M^+_i} \leq \gamma$.
Substituting into the above one gets
\[ \gamma\ge (1-\alpha_i)\gamma \ge \norm{  (1-\alpha_i)M_i^+ - (1-\alpha_i)U_{i-1}D_i} = \norm{D_{i+1}U_i -   \alpha_i I - (1-\alpha_i)U_{i-1}D_i  },
\]
so $X$ is a $\gamma$-eposet for $r_i=\alpha_i$ and $\delta_i=1-\alpha_i$.

\paragraph{(Proof of Part~\ref{parttwo})}
Suppose that $\norm{D_{i+1}U_i -   r_i I - \delta_i U_{i-1}D_i  }\leq \gamma$ for some $r_i,\delta_i\geq 0$.

We apply $D_{i+1}U_i  - r_i I - \delta_i U_{i-1}D_i$ to the constant vector $\one$, which is fixed by all of $D_{i+1}U_i ,I,U_{i-1}D_i$ because they are averaging operators. This gives
\[ \abs{1 - r_i - \delta_i} = \abs{\frac{\iprod{(D_{i+1}U_i  - r_i I - \delta_i U_{i-1}D_i) \one, \one} }{\iprod{\one, \one} }} \leq \gamma .\]
Next, we fix an arbitrary element $s\in X(i)$ and let $f=1_s$ be the function that equals $1$ on $s$ and $0$ elsewhere. Observe that $\iprod{M_i^+ f,f}=0$ so $\iprod{D_{i+1}U_i  f,f} = \alpha_i\iprod{f,f}$.
We apply $D_{i+1}U_i  - r_i I - \delta_i U_{i-1}D_i$ on the function $f = 1_s$,
\[ \abs{\alpha_i - r_i - \delta_i\frac{\iprod{U_{i-1}D_i f, f}}{\iprod{f,f}}} = \abs{\frac{\iprod{(D_{i+1}U_i  - r_i I - \delta_i U_{i-1}D_i) f, f} }{\iprod{f,f} }} \leq \gamma .\] Now, using the $\gamma$-laziness of the lower walks to bound $\iprod{U_{i-1}D_i f, f} \leq \gamma\iprod{f,f}$, we get $\abs{\alpha_i-r_i}\leq \gamma(1+\delta_i)$. Combining this with $\abs{\delta_i + r_i-1} \leq \gamma$ , we obtain that $\abs{\delta_i - (1-\alpha_i)} \leq \gamma (2+\delta_i)$. We can now lower- and upper-bound $\delta_i$ as follows. We have $\nicefrac{7}{8}\cdot \delta_i \leq (1-\gamma)\delta_i \leq 1 -\alpha_i + 2\gamma \leq 1 + \nicefrac14$. Hence, $\delta_i \leq \nicefrac{10}{7} \leq 2$. On the other hand, $\nicefrac{9}{8}\cdot \delta_i \geq (1+\gamma)\delta_i \geq 1-\alpha_i -2\gamma \geq \nicefrac12 -\nicefrac14 =\nicefrac14$. Hence, $\delta_i \geq \nicefrac29$.

Let us denote $A \approx B$ to mean $\norm{A-B} \leq O(\gamma)$. We have seen that $\delta_i + r_i \approx 1$ and that $\alpha_i \approx r_i$ (since $\delta_i \leq 2$), so
\begin{align*}
  \delta_i M_i^+ - \delta_i U_{i-1}D_i &\approx (1-\alpha_i)M_i^+ - \delta_i U_{i-1}D_i \\
  &= D_{i+1}U_i-\alpha_i I -\delta_i U_{i-1}D_i \\
  &\approx D_{i+1}U_i-r_i I -\delta_i U_{i-1}D_i,
\end{align*}
and we conclude that $\norm {M_i^+ - U_{i-1}D_i} = O(\gamma/\delta_i) = O(\gamma)$ (since $\delta_i \geq \nicefrac29$).
\end{proof}

Many posets satisfy the mild requirements of \cref{lem:eposets-defs}. For example, in the $(d+1)$-dimensional complete complex, the lower walk is $1/(n-d+1)$-lazy, and the upper walk is $(1,1/2,\ldots,1/d)$-uniformly lazy, and so $1/2$-uniformly lazy. Similarly, in the $(d+1)$-dimensional Grassmann complex, the lower walk is $\frac{q-1}{q^{n-d+1}-1}$-lazy, and the upper walk is $(1,\frac{q-1}{q^2-1},\ldots,\frac{q-1}{q^d-1})$-uniformly lazy, and so $1/(q+1)$-uniformly lazy.

In \cref{def:eposet}, we defined an \emph{$(\vec{r},\vec{\delta},\gamma)$-eposet} to be a poset where $\norm{DU-rI-\delta UD}\leq \gamma$. The above lemma states that this is equivalent to $\norm{U_{i-1}D_i - M^+_i} = O(\gamma)$ provided the
lower walk $UD$ is $\gamma$-lazy and the upper walk is $\nicefrac12$-uniformly lazy. This justifies the following equivalent definition of a $\gamma$-eposet.

\begin{definition}\label{def:alteposet}
  A $d$-dimensional poset $X$ is a $\gamma$-eposet if $\norm{U_{i-1}D_i - M^+_i} \leq \gamma$ for all $0\leq i< d$.
\end{definition}

\subsubsection{Decomposable posets}

The measured posets we consider in this section have a lattice-like property which we call \emph{decomposability}.

To define decomposabile posets, we first need the notion of a \emph{modular lattice}. Given a poset $X$ and elements $s_1,s_2 \in X$, the \emph{join} $s_1 \lor s_2$ of $s_1,s_2$ is an element $t$ such that $s_1,s_2 \leq t$, and $t \leq r$ whenever $s_1,s_2 \leq r$. If the join exists then it is unique. The \emph{meet} $s_1 \land s_2$ is defined analogously, with $\leq$ replaced by $\geq$. In a simplicial complex, join corresponds to union, and meet to intersection. A graded lattice is said to be \emph{modular} if $\rho(x) + \rho(y) = \rho(x \lor y) + \rho(x \land y)$ for all $x,y \in X$.
\begin{definition}[Decomposable measured posets]
\label{def:decomposable-measured-posets}
Let \(X\) be a measured poset. We say that \(X\) is \emph{decomposable} if the following conditions hold:
\begin{enumerate}
    \item $X$ is a modular lattice. In particular, $s_1,s_2 \in X(i)$ have a join in $X(i+1)$ iff they have a meet in $X(i-1)$.
    \item For any \(s_1,s_2 \in X(i)\) with meet $r \in X(i-1)$, it holds that
    \[\Pr_{M^+_i}[s_2 \sim s_1] = \Pr[\Dist{i-1}=r] \Pr_{M^{+,r}_0}[s_2 \sim s_1].\]
\end{enumerate}
\end{definition}

As an example, the Grassmann poset is decomposable. Indeed, it is well-known to be modular, since $\dim(s_1) + \dim(s_2) = \dim(s_1 + s_2) + \dim(s_1 \cap s_2)$. As for the second condition, it automatically holds whenever the non-lazy upper and lower walks coincide on all links, which is the case for the Grassmann poset. This is because the second condition is easily seen to hold if we replace the non-lazy upper walks with non-lazy down walks. It would be interesting to find other decomposable posets. One possible source, suggested by an anonymous referee, is the poset of flats of certain matroids. We leave this as a direction for future study.

Another way to obtain a decomposable measured poset is to start with one and introduce weights on the top level. This is a generalization of the special case of simplicial complexes, in which we consider an arbitrary distribution on the top facets. We describe this construction in detail in \Cref{sec:decomposable}.

We are now ready to state and prove the main theorem of this section,
\begin{theorem}[Equivalence of link-expansion and random-walk expansion] \label{thm:equivalence-eposets}
  Let \(X\) be a \(d\)-dimensional measured poset which is decomposable, the
lower walk $UD$ is $\gamma$-lazy and the upper walk is $\nicefrac12$-uniformly lazy.
  \begin{enumerate}
    \item If \(X\) is a \(\gamma\)-two-sided link expander, then \(X\) is a \(\gamma\)-eposet.
    \item If $X$ is a $\gamma$-eposet then $X$ is a $\eta^{-1}(1+\beta^{-1})\gamma$-two-sided link expander, where
  \[\eta = \min_{0 \leq i \leq d-2} \min_{r\in X(i), s\in X(i+1), r \leq s} \Pr[\Dist{i}=r \mid \Dist{i+1}=s]\]
  and
  \[\beta = 1 - \max_{0 \leq i \leq d-2} \max_{r\in X(i), s\in X(i+1), r \leq s}\Pr[\Dist{i}=r \mid \Dist{i+1}=s]\]
  \end{enumerate}
\end{theorem}

Before proving the theorem, let us calculate the values of $\eta$ and $\beta$ for simplicial complexes and for the Grassmann poset $\Gr_q(n,d)$.

For simplicial complexes, given that we chose $s \sim \Dist{i+1}$, the probability of choosing $r \in X(i), r \subseteq s$ is $\frac{1}{i+2}$. Thus $\eta = \frac{1}{i+2}, \beta = \frac{i+1}{i+2}$. Plugging this in \cref{thm:equivalence-eposets} recovers \cref{thm:equivalence}.

We continue with $X = \Gr_q(n,d)$. For \emph{any} $r \in X(i)$ and $s \in X(i+1)$ such that $r \leq s$ we have
\[
 \Pr[\Pi_i = r \mid \Pi_{i+1} = s] =
 \frac{1}{\sqbinom{i+2}{i+1}_q} = \frac{1}{1+q+\cdots+q^{i+1}}.
\]
This implies that $\eta = \frac{1}{1+q+\cdots+q^{d-1}}$ and $\beta = 1 - \frac{1}{1+q} = \frac{q}{1+q}$.

\begin{proof}
\textbf{Item 1.} Assume that $X$ is a $\gamma$-two-sided link expander. We show that
\[\norm{M^+_i - UD} \leq \gamma, \]
for all $i < d$.
Let $f$ be a function on $X(i)$, where $i < d$. %
By decomposability,
\begin{equation} \label{eq:localization-equation-eposets}
 \langle M^+_i f, f \rangle = \EE_{r \sim \Dist{i-1}}[ \langle M^{+,r}_0 f, f \rangle].
\end{equation}
We now note that for every \(r \in X(i-1)\),
\[
 \EE_{r \leq s}[f(s)] = (Df)(r).
\]
Therefore we have that
\begin{multline*}
\abs{\iprod{M^+_i f, f} - \iprod{UD f, f} } =  \left|\EE_{r \sim \Dist{i-1}} \EE_{s_1,s_2\sim M^{+,r}_0} [f(s_1) f(s_2)] - (Df)(r)^2\right| \leq
 \EE_{r\in \Dist{i-1}} \bigl[\lambda(M^{+,r}_0) \EE_{s \in X_r(0)}[f(s)^2]\bigr].
\end{multline*}
If $X$ is a $\gamma$-two-sided link expander then $\lambda(M^{+,r}_0) \leq \gamma$ for all $r$, and so
\[
 \abs{\iprod{(M^+_i - UD) f, f}} \leq \gamma \EE_{r \in \Dist{i-1}} \EE_{s \in X_r(0)} [f(s)^2] = \gamma \|f\|^2.
\]

\paragraph{Item 2.} Assume now that $X$ is a $\gamma$-HDX. Our goal is to show that for all $i < d-1$ and $r \in X(i)$,
\[ \lambda(M^{+,r}_0) \leq \eta^{-1}(1+\beta^{-1})\gamma.\]

Using the convention that $X(-1)$ consists of a single item $\emptyset$, for $i=-1$ we have $M^{+,\emptyset}_0 = M^+_0$, and so $U_{-1}D_{0}$ is zero on the space perpendicular to the constant function. Thus
\[\norm{M^+_0 - UD}  = \lambda(M^{+,\emptyset}_0), \]
and from our assumption $\lambda(M^{+,\emptyset}_0) \leq \gamma$.

Now assume $1 \leq i \leq d-1$, and fix some $r \in X(i-1)$. Let $f\colon X_r(0) \to \RR$ be some eigenfunction of $M^{+,r}_0$, which is perpendicular to the constant functions. In order to prove the theorem, we must show that
\[\abs{\frac{\iprod{M^{+,r}_0 f, f}}{\iprod{f,f}}} \leq \eta^{-1}(1+\beta^{-1}) \gamma.\]

Define a function $\tilde{f} \in C^{i}$ by
\[ \tilde{f}(s) = \begin{cases}
    f(s) &  \text{if }r \leq s, \\
    0 & \text{otherwise}.
  \end{cases} \]
Without loss of generality, we may assume that $\norm{\tilde{f}}=1$.

In order to obtain a bound on $\lambda(M^{+,r}_0)$, we bound $\iprod{\tilde{f},\tilde{f}}$, $\iprod{M^+_i \tilde{f}, \tilde{f}}$, and $\iprod{UD \tilde{f},\tilde{f}}$ in terms of $f$ and $M^{+,r}_0$.

We note that by Bayes' theorem,
\[\Pr[s \mid r] = \frac{\Pr[s]\Pr[r \mid s]}{\Pr[r]} \geq \frac{\eta}{\Pr[r]} \Pr[s]. \]
As \(f(s)=0\) whenever \(r\nleq s\),
\begin{equation}
  \label{eq:converse-1-eposet}
  \iprod{f,f} = \EE_{r\leq s}[f(s)^2] \geq \frac{\eta}{\Pr[r]} \EE_{s\in X(i)}[f(s)^2] = \frac{\eta}{\Pr[r]} \iprod{\tilde{f},\tilde{f}}.
\end{equation}

Furthermore, from what we shown in \eqref{eq:localization-equation-eposets} we obtain that
\[ \iprod{M^+_i \tilde{f}, \tilde{f}} = \EE_{r' \in X(i-1)}[\iprod{M^{+,r'}_{0} \tilde{f}, \tilde{f}}].\]

Fix some $r' \ne r$. If $\tilde{f} \ne 0$ on the link of \(r'\) then some $s \in X(i)$ satisifies both $r \leq s$ and $r'\leq s$; this $s$ must be the join of $r$ and $r'$, and so it is unique. Since $M^{+,r'}_{0}$ is a non-lazy operator, this implies that $\iprod{M^{+,r'}_{0} \tilde{f}, \tilde{f}} = 0$. We remain with
\begin{equation} \label{eq:converse-2-eposet}
  \iprod{M^+_i \tilde{f}, \tilde{f}} = \Dist{i-1}(r)\iprod{M^{+,r}_{0} f, f}.
\end{equation}
In other words, the upper non-lazy random walk is proportional to the local adjacency operator.

We now prove the following claim, which shows that the lower walk scales $\tilde{f}$ by a factor of at most $\eta^{-1}\gamma$:
\begin{claim}\label{f-down-operator-bounding-claim-eposet}
	If $f\colon X_r(0) \to \mathbb{R}$ is perpendicular to constant functions then
  	$|\iprod{U_{i-1}D_{i} \tilde{f}, \tilde{f}}| \leq \beta^{-1} \gamma$.
\end{claim}

Assuming the above:
 \begin{multline*}
 \abs{\frac{\iprod{M^{+,r}_0 f, f}}{\iprod{f,f}}} \leq |\eta^{-1} \Dist{i-1}(r)\iprod{M^{+,r}_0 f, f}| =
 \eta^{-1}|\iprod{M^+_i \tilde{f}, \tilde{f}}| \leq \\
 \eta^{-1}|\iprod{(M^+_i  - U_{i-1}D_{i})\tilde{f}, \tilde{f}}|  + \eta^{-1}|\iprod{U_{i-1}D_{i} \tilde{f}, \tilde{f}}| \leq \eta^{-1}(1+\beta^{-1}) \gamma,
 \end{multline*}
where the first line uses \eqref{eq:converse-1-eposet} and \eqref{eq:converse-2-eposet}, and the second line uses \cref{f-down-operator-bounding-claim-eposet}, our assumption that $\|M^+_i  - U_{i-1}D_{i}\| \leq \gamma$, and the triangle inequality.
\end{proof}

We complete the proof of \cref{thm:equivalence-eposets} by proving \cref{f-down-operator-bounding-claim-eposet}:
\begin{proof}[Proof of \cref{f-down-operator-bounding-claim-eposet}]
Since $UD$ is PSD, we have $\iprod{U_{i-1}D_i \tilde{f}, \tilde{f}} \geq 0$, and so we may remove the absolute value and prove
\[ \iprod{U_{i-1}D_{i} \tilde{f}, \tilde{f}} \leq \beta^{-1} \gamma.\]

Consider the inner product $\iprod{D_{i} \tilde{f}, D_{i} \tilde{f}}$ --- this is the expectation upon choosing $r' \sim \Dist{i-1}$, and then choosing two $i$-faces $s_1, s_2 \in X(i)$ containing it independently. Hence we decompose to the cases where $r' = r$ and $r' \ne r$:
\begin{gather}
\iprod{D_{i} \tilde{f}, D_{i} \tilde{f}} = \EE_{(r',s_1,s_2)}[f(s_1)f(s_2)] = \notag \\
\label{eq:down-op-1-eposet}
\Dist{i-1}(r)\EE_{(r',s_1,s_2)}[\tilde{f}(s_1)\tilde{f}(s_2) \mid r' = r] + (1-\Dist{i-1}(r))\EE_{(r',s_1,s_2)}[\tilde{f}(s_1)\tilde{f}(s_2) \mid r' \ne r].
\end{gather}
The first term is $0$, since from independence of $s_1, s_2$:
\[\EE_{(r',s_1,s_2)}[\tilde{f}(s_1)\tilde{f}(s_2) \mid r' = r] = \EE_{s_1}[f(s_1) \mid r \subset s_1]^2 = 0,\]
since by assumption $f$ is perpendicular to constant functions.

We saw above that for any $r' \ne r$, there is at most one $i$-face  $s'\geq r'$ such that $\tilde{f}(s) \ne 0$. For any $r' \ne r$, the value $\tilde{f}(s_1)\tilde{f}(s_2)$ is non-zero only when $s_1 = s_2=s'$.
For every $s_1 \in X(i)$, we define the event $E_{s_1}$ to hold when $s_2 = s_1$. Then
\[ \eqref{eq:down-op-1-eposet}=(1-\Dist{i-1}(r))\EE_{s_1}\bigl[\tilde{f}(s_1)^2 \Pr_{r', s_2}[E_{s_1} \mid r' \ne r] \bigr].\]
Note that if $s_1$ doesn't contain $r$ then $\tilde{f}(s_1)^2 = 0$, hence we continue taking expectation over all $s_1 \in X(i)$, even though some of them are unnecessary terms.

If we prove that for every $s \in X(i)$ so that $r \leq s$ we have $\Pr_{r', s_2}[E_{s_1} \mid r' \ne r] \leq \beta^{-1}\gamma$, then
\[(1- \Dist{i-1}(r))\EE_{s_1}[\tilde{f}(s_1)^2 \Pr_{r', s_2}[E_{s_1}] ] \leq \beta^{-1}\gamma \EE_{s_1}[\tilde{f}(s_1)^2] = \beta^{-1} \gamma \iprod{\tilde{f},\tilde{f}} = \beta^{-1} \gamma.\]
Thus we are left with proving the following statement: for all $s_1 \in X(i)$,
\[ \Pr_{r', s_2}[E_{s_1} \mid r' \ne r] \leq \beta^{-1}\gamma.\]

We first bound the unconditioned probability $\Pr[E_{s_1}] = \Pr_{r',s_2 \in X(i)}[s_2 = s_1 \mid r' \subset s_1,s_2]$. Fix some $s_1 \in X(i)$, and let $\one_{s_1} \colon X(i) \to \RR$ be its indicator. Notice that $U_{i-1}D_{i}\one_{s_1}(s_1) = \Pr_{r', s_2}[s_2 = s_1]$, and so
\[ \iprod{U_{i-1}D_{i} \one_{s_1}, \one_{s_1}} = \Dist{i}(s_1) U_{i-1}D_{i}\one_{s_1}(s_1) = \Dist{i}(s_1) \Pr[E_{S_1}]. \]
We again use the non-laziness property of $M^+_i$ to assert that $\iprod{M^+_i \one_{s_1}, \one_{s_1}} = 0$.
Since $X$ is a $\gamma$-eposet,
\[ \iprod{U_{i-1}D_{i} \one_{s_1}, \one_{s_1}} = \iprod{(U_{i-1}D_{i} - M^+_i) \one_{s_1}, \one_{s_1}} \leq \norm{U_{i-1}D_{i} - M^+_i} \|\one_{s_1}\|^2 = \gamma \Dist{i}(s_1). \]
Hence $\Pr[E_{s_1}] \leq \gamma$.

Consider now any $s_1 \in X(i)$ containing $r$ and some $r'\ne r$. The probability of sampling $r' \ne r$ given that we sample an element \(\leq s\) is at least $\beta$ (by definition of \(\beta\) which is the probability of \emph{not sampling} the element with largest probability), and so
\[ \Pr_{r', s_2}[E_{s_1} \mid r' \ne r]  \leq \beta^{-1} \Pr_{r', s_2}[E_{s_1}] \leq \beta^{-1} \gamma. \qedhere \]
\end{proof}

\subsubsection{Constructing decomposable posets}
\label{sec:decomposable}

Let $X$ be a measured poset given by the distribution $\vec{\Pi} = (\Pi_d,\ldots,\Pi_{-1})$. Given a distribution $\mathcal{D}$ on $X(d)$, we can construct a different distribution $\vec{\Psi}$ by first sampling $x \sim \mathcal{D}$, and then sampling $\vec{\Pi}$ conditioned on $\Pi_d = x$.

\begin{lemma} \label{lem:decomposability}
If $X$ is decomposable with respect to $\vec{\Pi}$, then it is decomposable with respect to $\vec{\Psi}$.	
\end{lemma}
\begin{proof}
Only the second condition depends on the measure. Let us spell it out in the case of the original distribution $\vec{\Pi}$	.

We are given $s_1,s_2 \in X(i)$ with meet $r \in X(i-1)$ and join $t \in X(i+1)$ (which is the only way for the upper walk to get from $s_1$ to $s_2$), and know that the following two expressions are equal:
\begin{gather*}
\Pr_{M_i^+}[s_2 \sim s_1] = \Pr[\Pi_i = s_1, \Pi_{i+1} = t] \frac{\Pr[\Pi_i = s_2 \mid \Pi_{i+1} = t]}{\Pr[\Pi_i \neq s_1 \mid \Pi_{i+1} = t]}, \\
\Pr[\Pi_{i-1} = r] \Pr_{M_0^{r,+}}[s_2 \sim s_1] =
\Pr[\Pi_{i-1} = r]
\Pr[\Pi_i = s_1, \Pi_{i+1} = t \mid \Pi_{i-1} = r] \frac{\Pr[\Pi_i = s_2 \mid \Pi_{i+1} = t, \Pi_{i-1} = r]}{\Pr[\Pi_i \neq s_1 \mid \Pi_{i+1} = t, \Pi_{i-1} = r]}.
\end{gather*}
Let us write these two expressions slightly differently:
\begin{gather*}
\Pr_{M_i^+}[s_2 \sim s_1] = \Pr[\Pi_{i+1} = t]\Pr[\Pi_i = s_1 \mid \Pi_{i+1} = t] \frac{\Pr[\Pi_i = s_2 \mid \Pi_{i+1} = t]}{\Pr[\Pi_i \neq s_1 \mid \Pi_{i+1} = t]}, \\
\Pr[\Pi_{i-1} = r] \Pr_{M_0^{r,+}}[s_2 \sim s_1] =
\Pr[\Pi_{i+1} = t]
\Pr[\Pi_i = s_1, \Pi_{i-1} = r \mid \Pi_{i+1} = t] \frac{\Pr[\Pi_i = s_2 \mid \Pi_{i+1} = t, \Pi_{i-1} = r]}{\Pr[\Pi_i \neq s_1 \mid \Pi_{i+1} = t, \Pi_{i-1} = r]}.
\end{gather*}
In order for $X$ to be decomposable with respect to $\vec{\Psi}$, we need these two expressions to coincide when replacing $\Pi$ with $\Psi$ throughout. Yet due to the definition of $\Psi$, this only replaces the $\Pr[\Pi_{i+1}=t]$ factors with $\Pr[\Psi_{i+1}=t]$ factors. Hence $X$ is decomposable with respect to $\vec{\Psi}$ as well.
\end{proof}

\subsection[Comparison to the Kaufman--Oppenheim Decomposition]{Comparison to the Kaufman-Oppenheim decomposition} \label{sec:decomposition-ko}
Kaufman and Oppenheim~\cite{KaufmanO2020} proposed a decomposition of \(C^k\) to orthogonal spaces in the case of high-dimensional expanders. Their definition extends to the general eposet setting:
\[ B^i = U^{k-i}C^i \cap \left ( \oplus_{j<i}B^j \right )^\perp = U^{k-i} C^i \cap \left( U^{k-(i-1)} C^{i-1} \right)^\perp.\]
(When $i = -1$, the definition is simply $B^{-1} = U^{k+1} C^{-1}$.)

As \(U=D^*\), we have an equivalent definition of these spaces by harmonic conditions similar to ours:
\[B^{i}=U^{k-i}C^i \cap \ker D^{k-i+1}.\]

By construction, these spaces are orthogonal, and it is easy to see that indeed their direct sum is \(C^k\).
Kaufman and Oppenheim~\cite[Theorem 1.5]{KaufmanO2020} showed that the subspaces $B^i$ are approximate eigenspaces of \(M^+\).

The following proposition shows that these two spaces are close.
\begin{proposition} \label{prop:KO-decomposition-is-close}
If \(f \in V^i\) has unit norm then there exists \(g\in B^i\) so that \(\norm{f-g} = O(\gamma)\).

Similarly, if \(g \in B^i\) has unit norm then there exists \(f\in V^i\) so that \(\norm{f-g} = O(\gamma)\).

The \(O\) notation may depend on \(k,r,\delta\) only.
\end{proposition}

\begin{proof}
We start with the first statement.
Let \(f \in V^i\) have norm \(1\). We decompose $f$ as \(f=\sum_{j=-1}^k g_j\), where $g_j \in B^j$, and take \(g=g_i\).
Then
\[ \norm{f-g}^2 = \iprod{f-g,f-g} = \iprod{f,f-g},\]
since \(g\) is perpendicular to \(f-g\). As \(f-g = \sum_{j=-1, j\ne i}^k g_j\), it is enough to show that \(\iprod{f,g_j} = O(\gamma)\) for all \(j\ne i\).

If \(j>i\) then by the definition of \(B^j\) and the fact that \(f\in U^{k-i}C^i\), \(\iprod{f,g_j} = 0\). If \(j<i\) then \(g_j \in U^jC^j = V^{-1} \oplus \cdots \oplus V^j\). By \cref{cor:approximate-orthogonality} and \cref{cor:approximate-l2-mass},
\(\iprod{f,g_j} = O(\gamma)(1+O(\gamma))\norm{f}\norm{g_j} = O(\gamma)\), since $\norm{g_j} \leq \norm{f} = 1$.

\smallskip
The proof of the second statement is similar.
Let \(g\in B^i\) have norm \(1\). We decompose $g$ as \(g=\sum_{j=-1}^k f_j\), where $f_j \in V^j$, and take \(f=f_i\). Then
\[ \norm{g-f}^2 = \iprod{g-f,g-f} = \iprod{g,g-f} + \iprod{f,g-f}.\]
Since \(g-f = \sum_{j=-1,j\ne i}^k f_j\), \cref{cor:approximate-orthogonality} and \cref{cor:approximate-l2-mass} show that \(\iprod{f,g-f} = O(\gamma)(1+O(\gamma)) \norm{f} \norm{g-f} = O(\gamma)\). Thus we need to show that \(\iprod{g,g-f} = O(\gamma)\), and for this it is enough to show that for every \(j\ne i\), \(\iprod{g,f_j} = O(\gamma)\).

If \(j<i\) then \(\iprod{g,f_j} = 0\) by the definition of \(B^i\) and the fact that \(f_j = U^{k-j}h\) for some $h \in C^j$. Otherwise \(j>i\), in which case we again use \cref{cor:approximate-orthogonality} and \cref{cor:approximate-l2-mass} to get the required bound.
\end{proof}

\begin{remark}
Let \(g \in B^i\), and let \(f \in C^i\) be a close vector promised in \cref{prop:KO-decomposition-is-close}.
Applying \cref{thm:approximate-orthogonality-eposets}, we get that
\[\norm{DUg - r_i^k g} \leq \norm{DU(g-f)} + \norm{DUf - r_i^k f} + |r_i^k|\norm{g-f} = O(\gamma).\]
In other words, $B^i$ is an approximate eigenspace of $DU$.
\end{remark}

\subsection{Decomposition in the Grassmann poset} \label{sec:grassmann}
Applying \cref{thm:decomposition-asd}, we obtain the following properties on the decomposition of $\Gr_q(n,d)$. These properties are well-known in the literature, but we rederive them to show the versatility of \cref{thm:approximate-orthogonality-eposets}:
\begin{claim} \label{clm:decomposition-Grassmann}
  Fix some $d,n \in \mathbb{N}$, let $X = \Gr_q(n,d)$, and let $\ell < d$. Let $f\colon X(\ell) \to \RR$ be an arbitrary function. Then we can decompose
  $f = f_{-1}+\cdots + f_\ell$, where $f_i \in V^i$:
  \begin{enumerate}
    \item For $i \ne j$, $\iprod{f_i, f_j} = 0$.
    \item $\norm{f}^2 = \norm{f_{-1}}^2+\cdots + \norm{f_{\ell}}^2$.
    \item The $f_i$'s are eigenvectors of $DU$. The eigenvalues are
        \[ r^\ell_i = 1 - \prod_{j=\ell-i+1}^\ell \left ( 1- \frac{q-1}{q^{j+2}-1} \right ) + \Theta(1/q^{n-\ell}).\]
    \item In particular, $DU$ has a constant spectral gap, that is, all its eigenvalues are bounded by a constant strictly smaller than $1$ when $n$ is large enough compared to $\ell$:
        \[ r^\ell_i \leq \frac{q}{q^2-1} + O \left( \frac{1}{q^{n-\ell}} \right) < 1 .\]
  \end{enumerate}
\end{claim}

\begin{proof}
The first two items are by invoking \cref{thm:approximate-orthogonality-eposets} and using \eqref{eq:Grassmann-SD}, which shows the Grassmann poset is a sequentially differential poset.

The third item is by invoking \cref{thm:approximate-orthogonality-eposets}, and using the fact that $\Gr_q(n,\ell)$ is also an \emph{expanding poset}, with $r_i = \frac{q-1}{q^{i+2}-1}$, $\delta_i = 1 - r_i$, and $\gamma = O(1/q^{n-\ell})$. The fourth item is by direct calculation: one may show using induction that the approximate formula for $r^\ell_i$ is
\[1 - \prod_{j=i}^\ell \left ( 1- \frac{q-1}{q^{j+2}-1} \right ) \leq  \sum_{j=i}^\ell \frac{q-1}{q^{j+2}-1}. \]
By taking $\ell$ to infinity and rearranging, we obtain
\[ \sum_{j=i}^\ell \frac{q-1}{q^{j+2}-1} \leq  (q-1) \sum_{j=i+2}^\infty \frac{q^{j}}{q^j(q^{j}-1)} \leq (q-1)\frac{q^{i+2}}{q^{i+2}-1} \sum_{j=i+2}^\infty \frac{1}{q^j}.\]
The infinite sum converges to $\frac{1}{q^{i+1}(q-1)}$, and so
\[ (q-1)\frac{q^{i+2}}{q^{i+2}-1}\sum_{j=i+2}^\infty \frac{1}{q^j} = (q-1)\frac{q^{i+2}}{q^{i+2}-1} \frac{1}{q^{i+1}(q-1)}  = \frac{q}{q^{i+2}-1} .\]
Hence $r^\ell_i \leq \frac{q}{q^2-1} + O \left( \frac{1}{q^{n-\ell}} \right)$.
\end{proof}

\begin{remark}
The actual values for $r^\ell_i$ can also be calculated by the formula devised in \cref{thm:approximate-orthogonality-eposets}. The calculations are omitted, as they don't add any additional insight.
\end{remark}
\begin{remark}
This result is also analogous to the decomposition of the complete complex, say the one obtained by Filmus and Mossel~\cite{FilmusM2019}.
\end{remark}
\subsection{Is there a bounded degree Grassmann poset?}
A high-dimensional expander, as constructed by Lubotzky, Samuels and Vishne~\cite{LubotzkySV2005-exphdx}, is a simplicial complex that is an eposet and a bounded-degree sub-complex of the complete complex. Is there an analogous construction of an eposet that is a bounded-degree subcomplex of the Grassmann poset? We conjecture the existence of such posets:
        \begin{conjecture}
            For any prime power $q$, $d \in \mathbb{N}$, and any $0 <\gamma < 1$, there exists an infinite sequence of natural numbers $n_1 < n_2 < n_3 < \dots$ such for all $n = n_j$ there exists a $d$-dimensional measured poset $X$ with the following properties:
            \begin{enumerate}
              \item $X$ is sparse, that is $|X(0)| = n$ and $X(d) = O(n)$ (the $O$-notation hides a constant that may depend on $q,d$, but not on $n$).
              \item $X$ may be embedded (as a poset) into $\Gr_q(n,d)$. In addition, for all $i < d$, $\Dist{i}$ is obtained by the same probabilistic experiment described for the Grassmann poset:
                \begin{enumerate}
            	   \item Choose a subspace of dimension $d+1$, $s_d \in X(d)$.
            	   \item Given a space $s_i$ of dimension $i+1$, choose $s_{i-1}$ to be a uniformly random codimension~$1$ subspace of $s_i$.
                \end{enumerate}
                In particular, $X$ is downward closed, that is, if $s\in X$ then every subspace $s'\subset s$ also belongs to $X$.
              \item $X$ is an $(\vec{r},\vec{\delta},\gamma)$-eposet for $r_i = \frac{1}{q^{i+2}-1}$ and $\delta_i = 1-r_i$.
            \end{enumerate}
        \end{conjecture}
        This existence of sub-posets as above is the vector-space analog of the existence of $\gamma$-HDX simplicial complexes. Moreover, it would be interesting to construct such a poset such that $\Dist{0}, \Dist{d}$ are uniform. Note however that even in the known constructions for $\gamma$-HDX simplicial complexes, $\Dist{d}$ is not uniform (but $\Dist 0$ is uniform).

Moshkovitz and Raz \cite{MoshkovitzR2008} gave a construction that can be viewed as an interesting step in this direction. They constructed, towards a derandomized low degree test, a small set of planes by choosing only planes spanned by directions coming from a smaller field $\mathbb{H} \subset \mathbb{F}_q$.

\subsection{Eposet parameters in a simplicial complex} \label{sec:eposet-as-simplicial-complexes}

Although the definition of (approximately) sequentially differential poset allows a range of parameters $\vec{r}$ and $\vec{\delta}$, these parameters turn out to be determined by the laziness of the upper walks, assuming that the lower walks are sufficiently non-lazy.
The lemma below shows that any family of simplicial complexes which are eposets, have parameters $\vec{r}, \vec{\delta}$ approaching $r_j = \frac{1}{j+2}$ and $\delta_j = 1-\frac{1}{j+2}$ as $\gamma$ goes to zero.

\begin{lemma}\label[lemma]{lem:eposet-fix-parameters}
Let $X^{(m)}$ be a sequence on $k$-dimensional $(\vec{r}^{(m)}, \vec{\delta}^{(m)}, \gamma^{(m)})$-eposets, where $\lim_{m \to \infty} \gamma^{(m)} = 0$. Then for all $j \leq k-1$:
\begin{equation} \label{eq:r-plus-delta-goes-to-one}
\lim_{m \to \infty} r_j^{(m)} + \delta_j^{(m)} = 1.
\end{equation}
Furthermore, suppose that the following two conditions hold:
\begin{enumerate}
\item For all $j \leq k-1$, the laziness of $U_{j-1}D_j$, goes to $0$ as $m$ goes to infinity:
\[ \lim_{m \to \infty} \Pr_{(t_1,t_2) \sim UD} [t_1 = t_2] = 0 .\]
\item There exists $\vec{\alpha}$ such that for all $j \leq k-1$, $D_{j+1}U_j = \alpha_j I + (1-\alpha_j) M^+$, where $M^+$ is a non-lazy averaging operator.
\end{enumerate}
Then
\[
\lim_{m \to \infty} r_j^{(m)} = \alpha_j \text{ and }
\lim_{m \to \infty} \delta_j^{(m)} = 1 - \alpha_j,
\]
\end{lemma}

In particular, if $X^{(m)}$ are $k$-dimensional simplicial complexes, then $\alpha_j = \frac{1}{j+2}$ and we get
\[
 \lim_{m \to \infty} r_j^{(m)} = \frac{1}{j+2} \text{ and }
 \lim_{m \to \infty} \delta_j^{(m)} = 1 - \frac{1}{j+2},
 \]
under the mild assumption that the laziness probability of $UD$ goes to zero. In other words, the interesting eposets are $\gamma$-HDXs.

\begin{proof}
To prove both assertions, we use the definition of an eposet to get the following inequality:
\begin{equation} \label{eq:asd-parameters-general}
  \abs{\frac{\iprod{(DU - r_j^{(m)} I - \delta_j^{(m)} UD)f, f} }{\iprod{f,f} }} \leq \gamma^{(m)},
\end{equation}
for any function $f \in C^k$. We use this inequality on specific functions $f$ we choose: the constant function, and indicator functions.

To show that
\[ \lim_{m \to \infty} r_j^{(m)} + \delta_j^{(m)} = 1,\]
we apply $DU - r_j^{(m)} I - \delta_j^{(m)} UD$ to the constant vector $\one$, which is fixed by all of $DU,I,UD$:
\[ \abs{\frac{\iprod{(DU - r_j^{(m)} I - \delta_j^{(m)} UD) \one, \one} }{\iprod{\one, \one} }} \leq \gamma^{(m)} \Longrightarrow
 |1 - r_j^{(m)} - \delta_j^{(m)}| \leq \gamma^{(m)},\]
thus $\lim_{m \to \infty} r_j^{(m)} + \delta_j^{(m)} = 1.$

To show that $\lim_{m \to \infty} r_j^{(m)} = \alpha_j$, we fix $j$ and take a sequence of $\sigma^{(m)} \in X(j)$ such that probability of laziness given that $t_1 = \sigma^{(m)}$ goes to zero:
\[\lim_{m \to \infty} \Pr_{(t_1,t_2) \sim UD}[ t_2 = \sigma^{(m)}  | t_1 = \sigma^{(m)}] = 0.\]
Denote by $\one_{\sigma^{(m)}}$ the indicator of $\sigma^{(m)}$. Then
\[ \frac{\iprod{UD \one_{\sigma^{(m)}}, \one_{\sigma^{(m)}}}}{\iprod{\one_{\sigma^{(m)}}, \one_{\sigma^{(m)}}}} = \Pr_{(t_1,t_2) \sim UD}[ t_2 = \sigma^{(m)}  | t_1 = \sigma^{(m)}].\]
Moreover,
\[ \frac{\iprod{(DU-r^{(m)}_jI) \one_{\sigma^{(m)}}, \one_{\sigma^{(m)}}}}{\iprod{\one_{\sigma^{(m)}}, \one_{\sigma^{(m)}}}} = \alpha_j - r_j^{(m)}. \]

Plugging $f = \one_{\sigma^{(m)}}$ into \eqref{eq:asd-parameters-general}, we get
\[ \bigl|\alpha_j - r^{(m)}_j - \delta^{(m)}_j \Pr_{(t_1,t_2) \sim UD}[ t_2 = \sigma^{(m)} | t_1 = \sigma^{(m)}]\bigr| \leq \gamma^{(m)}.\]
Since the $\delta^{(m)}_j$ are bounded, this shows that $\lim_{m \to \infty} r^{(m)}_j = \alpha_j$. The analogous statement for $\delta^{(m)}_j$ follows from~\eqref{eq:r-plus-delta-goes-to-one}.
\end{proof}

\section{Boolean degree 1 functions} \label{sec:exact-fkn}

In this section we characterize all Boolean degree 1 functions in nice complexes.

\begin{definition} \label{def:1-skeleton}
Let $X$ be a simplicial complex. The \emph{$1$-skeleton} of $X$ is the graph whose vertices are the $0$-faces of $X$ and whose edges are the $1$-faces of $X$.	
\end{definition}

\begin{theorem} \label{thm:hard-fkn}
 Suppose that $X$ is a proper $k$-dimensional simplicial complex, for $k \geq 2$, whose $1$-skeleton is connected. A function $f \in C^k$ is a Boolean degree~$1$ function if and only if there exists an independent set $I$ in the one skeleton of $X$ such that $f$ is the indicator of intersecting $I$ or of not intersecting $I$.
\end{theorem}
\begin{proof}
 If $f$ is the indicator of intersecting an independent set $I$ then $f = \sum_{v \in I} y_v$, and so $\deg f \leq 1$. If $f$ is the indicator of not intersecting an independent set $I$ then $f = \sum_{v \in X(0)} y_v/(k+1) - \sum_{v \in I} y_v$, and so again $\deg f \leq 1$.

 Suppose now that $f$ is a Boolean degree~$1$ function. If $|X(0)| \leq 2$ then the theorem clearly holds, so assume that $|X(0)| > 2$. \Cref{lem:degree-rep} shows that $f$ has a unique representation of the form
\[
 f = \sum_{v \in X(0)} c_v y_v.
\]
 Since $f$ is Boolean, it satisfies $f^2 = f$. Note that
\[
 f^2 = \sum_{\{u,v\} \in X(1)} 2c_u c_v y_{\{u,v\}} + \sum_{v \in X(0)} c_v^2 y_v.
\]
 Moreover, since every input $x$ to $f$ which contains $v$ contains exactly $k$ other points (elements of $X(0)$), and since $X(1)$ contains all pairs of points from $x$, we have
\[
 y_v = \sum_{u\colon \{u,v\} \in X(1)} \frac{y_{\{u,v\}}}{k}.
\]
 This shows that
\begin{multline*}
 0 = f^2 - f = \sum_{\{u,v\} \in X(1)} 2c_u c_v y_{\{u,v\}} + \frac{1}{k} \sum_{v \in X(0)} (c_v^2 - c_v) \sum_{u\colon \{u,v\} \in X(1)} y_{\{u,v\}} = \\
 \frac{1}{k} \sum_{\{u,v\} \in X(1)} (2kc_uc_v + c_u^2-c_u + c_v^2-c_v) y_{\{u,v\}}.
\end{multline*}
 \Cref{lem:degree-rep} shows that the coefficients of all $y_{\{u,v\}}$ must vanish, that is, for all $\{u,v\} \in X(1)$ we have
\[
 2kc_uc_v = c_u(1-c_u) + c_v(1-c_v).
\]
Consider now a triple of points $u,v,w$ such that $\{u,v,w\} \in X(2)$, and the corresponding system of equations:
\begin{align*}
 2kc_uc_v &= c_u(1-c_u) + c_v(1-c_v), \\
 2kc_uc_w &= c_u(1-c_u) + c_w(1-c_w), \\
 2kc_vc_w &= c_v(1-c_v) + c_w(1-c_w).
\end{align*}
 Subtracting the second equation from the first, we obtain
\[
 2kc_u (c_v - c_w) = c_v(1-c_v) - c_w(1-c_w) = (c_v-c_w)-(c_v^2-c_w^2) = (c_v-c_w)(1-c_v-c_w).
\]
 This shows that either $c_v = c_w$ or $2kc_u = 1-c_v-c_w$.

 If $c_u \neq c_v,c_w$ then $2kc_w + c_u + c_v = 2kc_v + c_u + c_w = 1$, which implies that $c_v = c_w$. Thus $c_u,c_v,c_w$ can consist of at most two values. If $c := c_u = c_v = c_w$ then $2kc^2 = 2c(1-c)$, and so $c \in \{0, \frac1{k+1}\}$. If $c := c_v = c_w \neq c_u$ then $2kc^2 = 2c(1-c)$, and so $c \in \{0, \frac1{k+1}\}$ as before. We also have $2kc_u c = c_u(1-c_u) + c(1-c)$. If $c = 0$ then this shows that $c_u (1-c_u) = 0$, and so $c_u = 1$. If $c = \frac1{k+1}$ then one can similarly check that $c_u = \frac1{k+1}-1$.

 Summarizing, one of the following two cases must happen:
\begin{enumerate}
 \item Two of $c_u,c_v,c_w$ are equal to $0$, and the remaining one is either $0$ or $1$.
 \item Two of $c_u,c_v,c_w$ are equal to $\frac1{k+1}$, and the remaining one is either $\frac1{k+1}$ or $\frac1{k+1}-1$.
\end{enumerate}

 Let us say that a vertex $v \in X(0)$ is of type A if $c_v \in \{0,1\}$, and of type B if $c_v \in \{\frac1{k+1}, \frac1{k+1}-1\}$. Since the complex is pure and at least two-dimensional, every vertex must participate in a triangle (two-dimensional face), and so every vertex is of one of the types. In fact, all vertices must be of the \emph{same} type. Otherwise, there would be a vertex $v$ of type A incident to a vertex $w$ of type B (since the link of $\emptyset$ is connected). However, since the complex is pure, $\{v,w\}$ must participate in a triangle, contradicting the classification above.

 Suppose first that all vertices are type A, and let $I = \{ v : c_v = 1 \}$. Note that $f$ indicates that the input face intersects $I$. Clearly $I$ must be an independent set, since otherwise $f$ would not be Boolean. When all vertices are type B, the function $1-f = \sum_{v \in X(0)} (\frac1{k+1} - c_v) y_v$ is of type A, and so $f$ must indicate not intersecting an independent set.
\end{proof}

If $X$ is a $\gamma$-HDX for $0 < \gamma < 1/(k+1)$ then the link of $\emptyset$ has positive spectral gap, and in particular it is connected. Thus \cref{thm:hard-fkn} applies to high-dimensional expanders.

When the $1$-skeleton of $X$ contains $r$ connected components $C_1,\ldots,C_r$, the same argument shows that the Boolean degree~$1$ functions on $X$ are of the form $f = f_1 + \cdots + f_r$, where each $f_i$ is the indicator of intersecting or not intersecting an independent set of $C_i$.

\section{FKN theorem on high-dimensional expanders} \label{sec:fkn}

In this section, we prove an analog of the classical result of
Friedgut, Kalai and Naor~\cite{FriedgutKN2002} for high-dimensional
expanders. The FKN theorem states that any Boolean function $F$ on
the hypercube that is close to a degree~1 function $f$ (not necessarily
Boolean) in the $L_2^2$-sense must agree with some Boolean degree~1
function (which must be a dictator) on most points. This result for the Boolean hypercube can be easily extended to functions on $k$-slices of the hypercube, provided $k = \Theta(n)$.

\begin{theorem}[FKN theorem on the slice~{\cite{Filmus2016-fkn}}]\label{thm:fkn-1}
  Let $n, k \in \integers_{\geq 0}$ and $\epsilon \in (0,1)$ such
  that $n/4 \leq k+1 \leq n/2$. Let $F\colon \binom{[n]}{k+1} \to\{0,1\}$
  be a Boolean function such that $\EE[(F-f)^2] < \epsilon$ for some
  degree~$1$ function $f\colon \binom{[n]}{k+1}\to\{0,1\}$. Then there exists
  a degree~$1$ function $g\colon \binom{[n]}{k+1} \to \RR$ such that
\[\Pr[F \neq g] = O(\epsilon).\]
 Furthermore, $g \in \{0,1,y_i,1-y_i\}$, that is, $g$ is a Boolean dictator ($1$-junta).
\end{theorem}

\begin{remark}\label[remark]{rem:fkn}
\begin{enumerate}
\item\label{item:fkn}
  The function $g$ promised by the theorem satisfies $\E[(g-F)^2] = \Pr[g \neq
  F] = O(\eps)$ and hence, by the $L_2^2$-triangle inequality we have
  $\E[(f-g)^2] \leq 2\E[(f-F)^2] + 2\E[(g-F)^2] = O(\eps)$. This is the way that the FKN theorem is traditionally stated,
  but we prefer the above formulation as this is the one we are able
  to generalize to the high-dimensional expander setting.
\item The function $1$ can also be written as $\frac{1}{k+1}\sum_j y_j$. The function $1-y_i$ can also be written as $\frac{1}{k+1}\sum_{j \neq i} y_j + (\frac{1}{k+1}-1) y_i$.
\item The result of Filmus~\cite{Filmus2016-fkn} is quite a bit stronger: for every $k \leq n/2$, it promises the existence of a function $g\colon \binom{[n]}{k+1} \to \RR$, not necessarily Boolean, such that $\E[(f-g)^2] = O(\eps)$. Moreover, either $g$ or $1-g$ is of the form $\sum_{i \in S} y_i$ for $|S| \leq \max(1,\sqrt{\epsilon} \cdot n/k)$. The bound on the size of $S$ ensures that $\Pr[g \in \{0,1\}] = 1 - O(\epsilon)$.
\end{enumerate}
\end{remark}

Our main theorem is an extension of the above theorem to $k$-faces of a two-sided link expander.

\begin{theorem}[FKN theorem for two-sided link expanders]\label{thm:fkn-hdx}
Let $X$ be a $d$-dimensional $\lambda$-two-sided link expander, where $\lambda < 1/d$, and let $4k^2 < d$. Let $F\colon X(k)\to\{0,1\}$ be a function such that
$\EE[(F-f)^2] < \epsilon$ for some degree~$1$ function $f\colon
X(k)\to\RR$. Then there exists a degree~$1$ function $g\colon X(k) \to
\RR$ such that \[\Pr[F \neq g] = O_\lambda(\epsilon).\] Furthermore, the degree~1 function $g$ can be written as $g(y) = \sum_i d_i y_i$, where $d_i \in \{0,1,\frac{1}{k+1},\frac{1}{k+1}-1\}$.
\end{theorem}

The high-dimensional analog of the FKN theorem is obtained from the
FKN theorem for the slice using the agreement theorem of Dinur and
Kaufman~\cite{DinurK2017}.

Using \cref{thm:equivalence}, we formulate the FKN theorem in terms of high-dimensional expanders:
\begin{corollary}[FKN theorem for HDX]
Let $X$ be a $d$-dimensional $\gamma$-high-dimensional expander, where $\gamma < 1/3d^2$, and let $4k^2 < d$. Let $F\colon X(k)\to\{0,1\}$ be a function such that
$\EE[(F-f)^2] < \epsilon$ for some degree~$1$ function $f\colon
X(k)\to\RR$. Then there exists a degree~$1$ function $g\colon X(k) \to
\RR$ such that \[\Pr[F \neq g] = O_\gamma(\epsilon).\] Furthermore, the degree~1 function $g$ can be written as $g(y) = \sum_i d_i y_i$, where $d_i \in \{0,1,\frac{1}{k+1},\frac{1}{k+1}-1\}$.
\end{corollary}

\subsection{Agreement theorem for high-dimensional expanders}

Dinur and Kaufman~\cite{DinurK2017} prove an agreement theorem for high-dimensional expanders. The setup is as follows. For each $k$-face $s$ we are given a local function $f_s\colon s\to\Sigma$ that assigns values from an alphabet $\Sigma$ to each point in $s$. Two local functions $f_s,f_{s'}$ are said to \emph{agree} if $f_s(v)=f_{s'}(v)$ for all $v\in s\cap s'$. Let $\calD_{k,2k}$ be the distribution on pairs $(s_1,s_2)$ obtained by choosing a random $t\sim \Dist{2k}$ and then independently choosing two $k$-faces $s_1,s_2 \subset t$. The theorem says that if a random pair  of faces $(s,s') \sim \calD_{k,2k}$ satisfies with high probability that $f_s$ agrees with $f_{s'}$ on the intersection of their domains, then there must be a global function $g\colon X(0)\to\Sigma$ such that almost always $g|_s \equiv f_s$. Formally:

\begin{theorem}[Agreement theorem for high-dimensional
  expanders~\cite{DinurK2017}]\label{thm:agree-hdx}
Let $X$ be a $d$-dimensional $\lambda$-two-sided high-dimensional expander, where $\lambda < 1/d$, let $k^2 < d$, and let $\Sigma$ be some fixed finite alphabet. Let $\{f_s\colon s
\to \Sigma\}_{s
  \in X(k)}$ be an ensemble of  local functions on
$X(k)$, i.e.\ $f_s \in \Sigma^s$ for each $s \in X(k)$. If
\[
\Pr_{(s_1,s_2)\sim \calD_{k,2k}}[f_{s_1}|_{s_1\cap s_2} \equiv
f_{s_2}|_{s_1 \cap s_2}]> 1 - \epsilon
\]
then there is a $g\colon X(0) \to \Sigma$ such that
\[
\Pr_{s\sim \Dist{k}} [ f_s \equiv g|_s ]\geq 1 - O_\lambda(\epsilon).\]
\end{theorem}

While Dinur and Kaufman state the theorem for a binary alphabet, the general version follows in a black box fashion by applying the theorem for binary alphabets $\lceil \log_2 |\Sigma| \rceil$ many times.

\subsection{Proof of \texorpdfstring{\cref{thm:fkn-hdx}}{FKN theorem}}

Let $f,F\in C^{k}$, where $F$ is a Boolean function and $f$ is a
degree~1 function, as in the hypothesis of \cref{thm:fkn-hdx}.
Since $f$ is a degree~1 function,
\cref{lem:degree-rep} guarantees that there exist $a_i \in \RR$ such that $f(y) =
\sum_{i \in X(0)} a_i y_i$.  Note that here we view the inputs of $f$ as
$|X(0)|$-bit strings with exactly $k+1$ ones, the rest being zero.

We begin by defining two ensembles of pairs of local functions $\{(f|_t,F|_t)\}_{t \in
  X(2k)}$, $\{(f|_u, F|_u)\}_{u \in X(4k)}$, which are the restrictions
of $(f,F)$ to the $2k$-face $t$ and $4k$-face $u$. Formally, for any $t \in X(2k)$ and $u\in X(4k)$, consider the restriction of
$f$ to $t$ and $u$ defined as follows:
\begin{align*}
f|_t, F|_t\colon \binom{t}{k} \to \RR, &&f|_t(y) = f(y) = \sum_{i \in
                                         t} a_i y_i, && F|_t(y) = F(y),\\
f|_u, F|_u\colon \binom{u}{k} \to \RR, && f|_u(y) = f(y) = \sum_{i \in
                                         u} a_i y_i, && F|_u(y)  = F(y).
\end{align*}
Observe that the $f|_t$'s are degree~1 functions, while the
$F|_t$'s are Boolean functions (similarly for $f|_u$'s and $F|_u$'s).

Now, define the following quantities:
\begin{align*}
\epsilon_t := \E_{s\colon s\subset t}[(f|_t(s)-
  F|_t(s))^2], &&\delta_u := \E_{s \colon s\subset u}[(f|_u(s)-
  F|_u(s))^2].
\end{align*}
Clearly, $\E_t[\epsilon_t],
\E_u[\delta_u] = O(\epsilon)$.

Let $\alpha_k = \frac1{k+1}$. Applying \cref{thm:fkn-1} (along with
\cref{rem:fkn}) to the functions $(f|_t, F|_t)$  for
each $t \in X(2k)$, we have the following claim:

\begin{claim} \label[claim]{lem:def-gt-1}
For every $t \in X(2k)$, there exists a Boolean dictator $g_t\colon \binom{t}{k}
\to \{0,1\}$ such that \[\E_{s \colon s \subset
  t}[(f|_t-g_t)^2] = O(\eps_t).\]
Furthermore, there exists a function $d_t\colon t \to \{0,1,\alpha_k,\alpha_k-1\}$
such that  $g_t(y) = \sum_{i \in t} d_{t}(i) y_i$.
\end{claim}

A similar claim holds for each $u \in X(4k)$:
\begin{claim} \label[claim]{lem:def-hu-1}
For every $u \in X(4k)$, there exists a Boolean dictator $h_u\colon \binom{u}{k}
\to \{0,1\}$ such that \[\E_{s \colon s \subset
  u}[(f|_u-h_u)^2] = O(\delta_u).\]
Furthermore, there exists a function $e_u\colon u \to \{0,1,\alpha_k,\alpha_k-1\}$
such that  $h_u(y) = \sum_{i \in u} e_{u}(i) y_i$.
\end{claim}

We now prove that functions in the collection of local functions $\{d_t\}_t$
typically agree with each other. This lets us use the
agreement theorem, \cref{thm:agree-hdx}, to sew these different local
functions together, yielding a single function $d\colon X(0) \to
\{0,1,\alpha_k, \alpha_k -1\}$. This $d$ determines a global
degree~1 function $g$ defined as follows: $g(y) = \sum_{i \in X(0)}
d(i)y_i$.

\begin{claim} \label[claim]{lem:def-d}
There exists a function $d\colon X(0) \to  \{0,1,\alpha_k, \alpha_k
-1\}$ such that  $\Pr_{t}[d_t \equiv d|_t] = 1 - O_\lambda(\epsilon)$.
\end{claim}
\begin{proof}
To sew the various $d_t$ together via the agreement theorem, we would like to first bound the probability
\[\Pr_{(t_1,t_2) \sim \calD_{2k,4k}}[d_{t_1}|_{t_1 \cap t_2} \not\equiv
d_{t_2}|_{t_1 \cap t_2}]\;.\]
Recall the definition of the
distribution $\calD_{2k,4k}$: we first pick a set $u \in X(4k)$
according to $\Dist{4k}$ and then two $2k$-faces $t_1,t_2$ of
$u$ uniformly and independently. Consider the three functions $d_{t_1}, d_{t_2}$ and
$e_u$. Clearly, if $d_{t_1}|_{t_1 \cap t_2} \not\equiv d_{t_2}|_{t_1\cap t_2}$ then one of
$e_{u}|_{t_1} \not\equiv d_{t_1}$ or $e_{u}|_{t_2} \not\equiv d_{t_2}$ must hold. Thus,
\begin{equation}\label{eq:dt_vs_eu}
\Pr_{{(t_1,t_2) \sim \calD_{2k,4k}}}[d_{t_1}|_{t_1 \cap t_2} \not\equiv d_{t_2}|_{t_1 \cap t_2}] \leq
 2\cdot \Pr_{t,u}[e_{u}|_t \not\equiv d_t]\;.
\end{equation}
Thus, it suffices to bound the probability $\Pr_{t,u}[e_{u}|_t \not\equiv
d_t]$, where $u \sim \Dist{4k}$ and $t$ is a random $2k$-face
of $u$.

For any fixed $t \subset u$, the $L_2^2$ triangle inequality shows that
\[
 \EE[(h_u|_t-g_t)^2] \leq 2\EE[(h_u|_t-f|_t)^2] + 2\EE[(f|_t-g_t)^2] =
 2\EE[(h_u|_t-f|_t)^2]+ O(\epsilon_t).
\]
 Taking expectation over $t\in X(2k)$ conditioned on $t \subset u$, we see that
\[
 \EE_{t \subset u} \EE[(h_u|_t-g_t)^2] \leq
 2\EE[(h_u - f|_u)^2] + O\left(\EE_{t \colon t\subset u}\epsilon_t\right) = O(\delta_u) + O\left(\EE_{t \colon t\subset u}\epsilon_t\right).
\]
Taking expectation over $u \sim \Dist{4k}$, we now have
\[ \EE_u \EE_{t \subset u}\EE[(h_u|_t-g_t)^2] = O(\eps). \]
For any fixed $t \subset u$, both $h_u|_t$ and $g_t$ are Boolean dictators. Hence either they agree, or $\EE[(h_u|_t-g_t)^2] = \Omega(1)$. This shows that $h_u|_t$ disagrees with $g_t$ with probability $O(\eps)$, and so
\[
 \Pr_{t,u}[e_{u|_t} \not\equiv d_t] = O(\eps).
\]
We now return to \eqref{eq:dt_vs_eu}, concluding that
\[
 \E_{(t_1,t_2)\sim\calD_{2k,4k}}\left[d_{t_1}|_{t_1 \cap t_2} \not\equiv
   d_{t_2}|_{t_1 \cap t_2})\right] =O(\eps).
\]

We have thus satisfied the hypothesis of the agreement theorem
(\cref{thm:agree-hdx}). Invoking the agreement theorem, we deduce
that $\Pr_{t \sim \Dist{2k}}[d_t \equiv d|_t] = 1 - O_\lambda(\epsilon)$.
\end{proof}

The $d$'s guaranteed by \cref{lem:def-d} naturally correspond to
a degree~1 function $g\colon X(k) \to \RR$ as follows:
\[
g(y) : = \sum_{i \in X(0)} d(i) y_i.
\]

We now show that this $g$ is mostly Boolean.
\begin{claim}\label[claim]{clm:gbool} $\Pr_s[g(s) \in \{0,1\}] = 1- O_\lambda(\eps)$.
\end{claim}
\begin{proof}
Since
$g_t$ is Boolean-valued,
\[
 \Pr_{s \sim \Dist{k}}[g(s) \in \{0,1\}] \geq \Pr_t[g|_t = g_t] = \Pr_t[d|_t \equiv d_t] = 1 - O_\lambda(\epsilon). \qedhere
\]
\end{proof}

We now show that $g$ in fact agrees pointwise with $F$ most of the time.
\begin{claim}\label[claim]{clm:geqF} $\Pr_s[g \neq F] = O_\lambda(\eps)$.
\end{claim}
\begin{proof}
Fix any $t \in X(2k)$. We compute $\Pr_{s \colon s \subset t}[F|_t
\neq g_t]$ as follows
\begin{align*}
\Pr[F|_t \neq g_t] & = \|F|_t - g_t\|^2
  &[\text{ Since $F|_t$ and $g_t$ are both Boolean }]\\
&\leq 2\cdot \|F|_t - f|_t\|^2 + 2 \cdot \|f|_t - g_t\|^2\\
& = O(\eps_t) + O(\eps_t) = O(\eps_t).
\end{align*}
We can now compute $\Pr_s[F \neq g]$ as follows:
\[
\Pr[F \neq g] = \E_t \Pr[F|_t \neq g|_t] \leq
\E_t \Pr[F|_t \neq g_t] + \Pr_t[g|_t \neq
g_t] = O(\eps) + \Pr_t[d|_t \not\equiv d_t] = O_\lambda(\eps). \qedhere
\]
 \end{proof}

This completes the proof of \cref{thm:fkn-hdx}.

{\small
\bibliographystyle{prahladhurl}
\bibliography{DDFH-hdx-bib}
}
\end{document}